\documentclass[lettersize,journal]{IEEEtran}
\usepackage{array}
\usepackage[caption=false,font=normalsize,labelfont=sf,textfont=sf]{subfig}
\usepackage{stfloats}
\usepackage{verbatim}

\usepackage{cite}
\usepackage{bm}
\usepackage{amsmath,amssymb,amsfonts}
\usepackage{algorithm}
\usepackage{algorithmic}
\usepackage{graphicx}
\usepackage{textcomp}
\usepackage{xcolor}
\usepackage{gensymb}
\usepackage{comment}
\usepackage{soul}
\usepackage{color, xcolor}
\usepackage{amsthm}
\usepackage{bm}
\usepackage{multicol}
\usepackage{enumitem}
\usepackage{booktabs}
\usepackage{array}

\theoremstyle{plain}

\newtheorem{lemma}{Lemma}
\newtheorem{theorem}{Theorem}
\newtheorem{remark}{Remark}
\newtheorem{proposition}{Proposition}
\newtheorem{corollary}{Corollary}
\newcolumntype{C}[1]{>{\centering\arraybackslash}p{#1}}
\newcolumntype{M}[1]{>{\centering\arraybackslash}m{#1}}

\def\BibTeX{{\rm B\kern-.05em{\sc i\kern-.025em b}\kern-.08em
		T\kern-.1667em\lower.7ex\hbox{E}\kern-.125emX}}
\usepackage{balance}

\begin{document}
	\title{Spectral Approximation and Ergodic-Capacity Convergence of HMIMO Channels under Spatial–Wavenumber Domain Mismatch}
	
	\author{Hangsong Yan, \textit{Member}, \textit{IEEE}, Hong Yang, \textit{Senior Member}, 
	\textit{IEEE}, Shu Sun, \textit{Senior Member}, \textit{IEEE}
	
		\thanks{
			Hangsong Yan is with the Hangzhou Institute of Technology, Xidian University, Hangzhou, China (email: yanhangsong@xidian.edu.cn).
		
			Hong Yang (retired) was with the Department of Mathematics and Algorithms Research, Nokia Bell Labs, Murray Hill, USA (email: hyang.bell.labs@gmail.com). 
			
			Shu Sun is with the School of Information Science and Electronic Engineering, Shanghai Jiao Tong University, Shanghai, China (email: shusun@sjtu.edu.cn).

		}}
	
	
\maketitle

\begin{abstract}
	We establish quantitative results on finite-dimensional spectral approximation and ergodic-capacity convergence for continuous Holographic Multiple-Input Multiple-Output (HMIMO) channels in the square-aperture setting with physically prescribed circular wavenumber support.
	The resulting spatial-wavenumber domain mismatch leads to a non-separable square–disk concentration problem for which the classical separable construction based on prolate spheroidal wave functions (PSWFs) cannot be directly applied.
	Specifically, we project the continuous operator onto a tensor-product subspace of one-dimensional (1D) PSWFs while preserving the circular wavenumber support, yielding a generally non-diagonal but highly sparse finite-dimensional matrix.
	We show that the whole-spectrum approximation error, jointly accounting for retained-eigenvalue perturbations and the residual spectral tail, remains controlled by a 1D PSWF eigenvalue-tail envelope despite the loss of separability and induced off-diagonal coupling. Beyond an explicit 1D truncation threshold, this error decays at a certified super-exponential rate.
	This analysis further yields an explicit asymptotic upper envelope for the eigenspectrum under the flattened two-dimensional eigenvalue ordering. 
	We further establish an explicit non-asymptotic upper bound on the gap between the actual ergodic capacities of the continuous and tensor-PSWF-truncated channels under their respective transmit-covariance optimizations. Combined with the spectral result, this capacity-gap bound inherits the same super-exponential dependence on the truncation order. Finally, quadrature rules with explicit radial and angular node thresholds are developed for numerical evaluation of the projected matrix. 
	Numerical results reveal the advantage of the proposed finite-dimensional analysis over conventional truncation based on spatial degrees of freedom in retaining performance-relevant modes, particularly for compact apertures.
\end{abstract}
	

\section{Introduction}
\label{sec:Introduction}
The evolution of wireless communications towards Holographic Multiple-Input Multiple-Output (HMIMO) has motivated a shift from discrete matrix channels to continuous spatial electromagnetic operators~\cite{Wei2026EIT,Hu2026Electromagnetic}. Unlike conventional MIMO systems, where spatial dimensions are represented by discrete antenna indices, HMIMO relies on propagating electromagnetic fields generated and received over spatially continuous apertures. Consequently, from an information-theoretic perspective, the physical channel is no longer modeled as a finite-dimensional matrix, but rather as a continuous spatial-wavenumber integral operator. Evaluating the fundamental capacity limits of such channels therefore requires characterizing the eigenspectrum of the underlying continuous operator.

Historically, the spectral analysis of bandlimited integral operators traces back to the seminal concentration theory developed in~\cite{Slepian1961Prolate}, where prolate spheroidal wave functions (PSWFs) arise as the eigenfunctions of the classical one-dimensional time--frequency limiting problem. This theory provides a rigorous mathematical foundation for characterizing the spatial dimensionality of one-dimensional bandlimited continuous channels. A systematic exposition of the analytical properties of PSWFs is provided in~\cite{osipov2013prolate}.
PSWFs have also been applied to HMIMO and related electromagnetic information-theoretic problems~\cite{Liang2024Achievable, zhu2026mimocapacity}. The authors of~\cite{Liang2024Achievable} analyzed the achievable rate of linear HMIMO systems for both normal and non-normal additive white Gaussian noise channels. Separately, the authors of~\cite{zhu2026mimocapacity} established a discrete-continuous correspondence based on PSWFs for linear HMIMO systems and proposed a PSWF-based channel estimator.

Extensions of the classical Slepian concentration problem to higher dimensions give rise to generalized PSWFs (GPSWFs), introduced by Slepian~\cite{Slepian1964Prolate}.
Related higher-dimensional spatial--spectral concentration problems have subsequently been studied for several symmetric geometries, including the sphere~\cite{simons2006spatiospectral} and the three-dimensional ball~\cite{KHALID2016470}. More recently, \cite{greengard2024generalized} developed algorithms and analysis for GPSWFs in arbitrary dimensions, including efficient evaluation, eigenvalue computation, quadrature rules, and interpolation formulae. Furthermore, \cite{ZHANG2020539} introduced ball PSWFs on the unit ball in arbitrary dimensions as a generalization of orthogonal ball polynomials. For the symmetry-based GPSWF constructions discussed above, analytical tractability is closely associated with geometries that admit a commuting differential operator, such as spherical or ball domains~\cite{simons2006spatiospectral,ZHANG2020539}. A square aperture coupled with the inherently circular wavenumber support dictated by wave propagation, however, does not possess the same rotationally symmetric and separable structure. Consequently, the classical commuting Sturm--Liouville formulations underlying these analytically tractable GPSWF constructions cannot be directly applied to the square--disk configuration.

Beyond the symmetry-based GPSWF constructions discussed above, the Slepian concentration problem has also been studied for more general spatial and spectral geometries. The two-dimensional Cartesian formulation in~\cite{simons2011cartesian} allows the spatial and spectral concentration regions to have arbitrary geometries in principle. Quantitative spectral analysis was subsequently developed in~\cite{israel2024eigenvalue} for hypercube spatial domains and coordinate-wise symmetric convex spectral domains, with bounds on eigenvalue counting, the plunge region, and eigenvalue decay. The analysis in~\cite{marceca2024eigenvalue} extended quantitative eigenvalue estimates to substantially broader regular spatial and spectral domains and further studied finite discretizations of the corresponding Fourier concentration operators. For bounded spatial and spectral domains with maximally Ahlfors-regular boundaries, the results in~\cite{hughes2025eigenvalueII} further sharpened the quantitative characterization of the plunge region and eigenvalue distribution. More recently, a disk-adapted wave-packet frame was constructed in~\cite{hughes2026wavepackets} for a disk spatial domain and a broad class of well-shaped spectral domains, yielding improved eigenvalue estimates for this geometry.

While the above mathematical studies provide increasingly general and quantitative characterizations of multidimensional concentration operators, they primarily concern the mathematical analysis of the underlying concentration problem, including its formulation, eigenvalue distribution, geometry-adapted representations, and numerical discretization. 
From the HMIMO perspective, conventional spatial degrees of freedom (DoF) characterizations provide a macroscopic measure of effective spatial dimensionality, but do not quantify how accurately a given finite-dimensional truncation captures the eigenspectrum of the underlying continuous operator. This raises a more refined spectral question: whether, in the presence of spatial–wavenumber domain mismatch and the resulting loss of separability, a structured finite-dimensional representation can be constructed whose whole-spectrum approximation accuracy can be rigorously quantified as a function of truncation order.
A further information-theoretic question is whether the actual ergodic-capacity loss induced by this finite-dimensional truncation can be rigorously quantified under transmit-covariance optimization. 

To address these questions in the square-aperture setting, we exploit the well-established analytical properties of one-dimensional (1D) PSWFs and construct a two-dimensional (2D) projection framework from their tensor products while preserving the physically prescribed circular wavenumber support. Specifically, our main contributions are summarized as follows:
\begin{itemize}
	\item
	We develop a structured finite-dimensional representation of the non-separable square--disk concentration operator by projecting it onto a tensor-product subspace constructed from 1D PSWFs, while retaining the physical circular wavenumber support. The resulting projected matrix is generally non-diagonal, thereby preserving the coupling induced by the square--disk mismatch. We derive an explicit analytical expression for its matrix elements and show that the parity structure induces exact sparsity, substantially reducing the number of elements that require evaluation.
	
	\item
	We rigorously quantify the spectral approximation error over the entire spectrum, jointly accounting for retained-eigenvalue perturbations and the residual spectral tail of the continuous operator. Using the max--min principle for compact self-adjoint operators, we establish an exact trace-defect representation of this total error. Despite the loss of separability and the resulting non-diagonal coupling, we further show that this trace defect can be controlled by a 1D PSWF eigenvalue tail envelope. This yields an explicit non-asymptotic super-exponential upper bound beyond an explicit truncation threshold. The resulting bound further yields an explicit asymptotic upper envelope for the eigenvalues of the continuous operator under the flattened 2D eigenvalue ordering. Moreover, the same trace-defect argument can be extended to rectangular apertures with unequal side lengths.
	
	\item
	We establish an explicit non-asymptotic upper bound on the gap between the actual ergodic capacities of the continuous and tensor-PSWF-truncated channels under perfect instantaneous channel state information at the receiver (CSIR) and statistical channel state information at the transmitter (CSIT). The bound shows that the capacity loss under transmit-covariance optimization is controlled by the same trace defect governing the whole-spectrum approximation error. Combined with the spectral result above, the capacity-gap bound inherits the same super-exponential dependence on the truncation order.
 
	\item
	For the accurate numerical evaluation of the eigenspectrum, we propose quadrature rules for computing the elements of the truncated sparse matrix. Specifically, we derive explicit quadrature-node thresholds for both the radial and angular integrals. Beyond these thresholds, the radial and angular quadrature errors exhibit super-exponential and exponential convergence, respectively. 
	Numerical results further reveal that, particularly for compact apertures, truncating the 1D PSWF basis at the conventional 1D spatial DoF benchmark can omit modes that make non-negligible contributions to the evaluated spectral efficiency.
\end{itemize}

The organization of this paper is as follows.
Section~\ref{sec:Continuous small-scale fading model} formulates the physical constraints of continuous HMIMO channels and establishes the mathematical origin of the spatial--wavenumber domain mismatch. Section~\ref{sec:linear aperture} reviews the spectral properties and numerical evaluation of 1D PSWFs for the linear aperture. Section~\ref{sec:two dimensional case} develops the tensor-PSWF spectral approximation framework for the square aperture, deriving the whole-spectrum approximation bound and the resulting eigenvalue upper envelope. Section~\ref{sec:channel_capacity} establishes the ergodic capacities of the continuous and tensor-PSWF-truncated channels and derives an explicit bound on the corresponding capacity gap. Section~\ref{sec:Quadrature Rules and Numerical Results} develops the quadrature rules for matrix-element computation and presents the numerical results. Finally, Section~\ref{sec:conclusion} concludes the paper and discusses future research directions.

\textbf{Notation}: Unless otherwise stated, scalar quantities are denoted by italic symbols, while boldface lowercase and uppercase symbols denote vectors and matrices, respectively. Calligraphic uppercase symbols denote linear operators. The superscripts $(\cdot)^T$ and $(\cdot)^{\dagger}$ denote the transpose and adjoint, respectively, where the latter reduces to the conjugate transpose for finite-dimensional matrices. The operators $\mathbb{E}[\cdot]$, $\operatorname{Tr}(\cdot)$, and $\operatorname{diag}(\cdot)$ denote expectation, trace, and the diagonal matrix/operator formed by its arguments, respectively. The notations $\|\cdot\|$ and $\|\cdot\|_{\mathrm{HS}}$ denote the operator and Hilbert-Schmidt norms, respectively. The notation $\mathcal{CN}(0, 1)$ represents a zero-mean proper complex Gaussian distribution with unit variance. We adopt the Fourier transform conventions $\mathcal{F}[f(x)](k_x) = \int_{-\infty}^{\infty}f(x)e^{-ik_x x} dx$ and $\mathcal{F}[f(x, y)](k_x, k_y) = \int_{-\infty}^{\infty}\int_{-\infty}^{\infty}f(x,y)e^{-i(k_x x + k_y y)} dxdy$.
	
\section{Physical Constraints of Continuous HMIMO Channels}
\label{sec:Continuous small-scale fading model}
We begin with the monochromatic, far-field, continuous small-scale fading, $h$, in a source-free environment. Utilizing the plane-wave representation from \cite{Pizzo2020Spatially}, the fading is described as a scalar random field fundamentally constrained by the physical laws of wave propagation \cite{Paulraj2003Introduction}. The scalar Helmholtz equation in the wavenumber domain \cite{hildebrand1962advanced} implies the following relationship between the wavenumber $\kappa = 2\pi/\lambda$ (where $\lambda$ is the wavelength) and the Cartesian components of the wavenumber vector $(k_x, k_y, k_z)$:
$$
k_x^2 + k_y^2 + k_z^2 = \kappa^2.
\label{eq:k_x_k_y_k_z}
$$
This physical constraint yields two solutions for the vertical component $k_z$:
$$
k_z = \pm\sqrt{\kappa^2 - k_x^2 - k_y^2}.
\label{eq:k_z}
$$
The two solutions correspond to upgoing $(+)$ and downgoing $(-)$ propagating waves, and the total field at a spatial point $(x,y,z)$ is the superposition of these two components: $h(x, y, z) = h_{+}(x, y, z) + h_{-}(x, y, z)$. Each component is expressed through a plane-wave spectral representation~\cite{Pizzo2020Spatially} as follows:
\begin{align}
	h_{\pm}(x,y,z) = &\frac{1}{4\pi\sqrt{\pi}}\iint_{\mathcal{D}(\kappa)}\frac{A_{h,\pm}(k_x,k_y)}{(\kappa^2 - k_x^2 - k_y^2)^{1/4}}W^{\pm}(k_x, k_y) \nonumber \\
	&\qquad \times e^{i(k_x x + k_y y \pm \sqrt{\kappa^2 - k_x^2 -k_y^2}z)}\,dk_x\,dk_y,
	\label{eq:h_pm_x_y_z}
\end{align}
where $A_{h,\pm}$ are real-valued non-negative deterministic fields called spectral factors, and $W^{\pm}$ are two independent 2D, zero-mean, complex-valued, spatially white Gaussian random fields with unit variance. 
	
Because our analysis is restricted to the far-field regime, we can neglect the effects of evanescent waves (corresponding to the case where $k_z^2 < 0$), as they decay exponentially with distance and become negligible far from the source. This physical constraint confines the integration domain in \eqref{eq:h_pm_x_y_z} to the circular disk $\mathcal{D}(\kappa)$:
\begin{equation}
	\mathcal{D}(\kappa) = \{(k_x, k_y) \in \mathbb{R}^2: k_x^2 + k_y^2 \leq \kappa^2 \}.
	\label{eq:D_kappa}
\end{equation}
The domain defined in \eqref{eq:D_kappa} dictates that the considered small-scale fading is band-limited with a circular support in the wavenumber domain. This geometric property serves as the physical origin of the spatial-wavenumber domain mismatch addressed in the subsequent analysis.
	
\section{Linear Aperture}
\label{sec:linear aperture}
We begin by reviewing the linear aperture case. The spectral properties of a 1D continuous aperture provide the foundational building blocks for analyzing the 2D domain-mismatched continuous operator in subsequent sections.
We consider the 1D small-scale fading along the linear array. Let $s \in [0, L_x]$ denote the physical spatial coordinate, where $L_x$ is the physical aperture length. 
Based on the channel model established in Section \ref{sec:Continuous small-scale fading model}, the fading field is band-limited in the wavenumber domain, with spectral support confined to $[-\kappa, \kappa]$, while it is observed over the finite spatial aperture $s\in [0, L_x]$.
This setup corresponds to the classic Slepian concentration problem~\cite{Slepian1961Prolate} of finding band-limited functions whose energy is maximally concentrated within a finite interval.
	
With a slight abuse of notation, we reuse the variable $x$, previously denoting the absolute physical coordinate, to represent the standard normalized interval $[-1, 1]$, which is mapped from the physical domain $s \in [0, L_x]$ via an affine transformation. 
Under this setting, the optimal functions that maximize the energy concentration within $[-1, 1]$ among all strictly band-limited functions are the eigenfunctions of the following integral equation~\cite{Slepian1961Prolate}:
\begin{equation}
	F_c[\phi_n](x) \triangleq \int_{-1}^{1}e^{icxx'}\phi_n(x')dx' = \lambda_n \phi_n(x),
	\label{eq:PSWF_definition}
\end{equation}
where $x, x' \in [-1, 1]$, $\{\lambda_n\}$ are the eigenvalues of the integral operator $F_c$, and $c = 2\pi W T$ is the time-bandwidth product. This optimality arises because the eigenfunctions $\{\phi_n(x)\}$ maximize the energy concentration ratio in the normalized interval $[-1, 1]$ (analogous to the time interval $[-T, T]$) among all functions that are strictly band-limited. Thus, for a given number of 1D basis functions, any such signal can be approximated by its orthogonal projection onto the subspace spanned by $\{\phi_n(x)\}_{n=0}^{N_{1D}-1}$, which minimizes the $[-1, 1]$ approximation error within this subspace.
We also note that throughout this paper, $\phi_n(x)$ denotes the globally defined bandlimited PSWF whose restriction to $x \in [-1, 1]$ satisfies the finite Fourier eigenvalue equation \eqref{eq:PSWF_definition}. In particular, its values outside $[-1, 1]$ are defined by the bandlimited continuation of~\eqref{eq:PSWF_definition}.
	
In our considered physical model, the time interval $[-T, T]$ and the angular frequency band $[-2\pi W, 2\pi W]$ correspond to the spatial aperture interval $[0, L_x]$ and the wavenumber band $[-\kappa, \kappa]$, respectively. 
Under the spatial normalization from $[0, L_x]$ to $[-1, 1]$, the corresponding wavenumber variables are scaled by $L_x/2$. For notational simplicity, we continue to use the same wavenumber notation, so that the physical band $[-\kappa, \kappa]$ is mapped to the normalized band $[-c, c]$, where $c=\frac{\kappa L_x}{2}
=\frac{\pi L_x}{\lambda}$.
\begin{remark}
	For a linear HMIMO array of length $L_x$, the conventional spatial DoF are characterized by $\frac{2L_x}{\lambda}$~\cite{pizzo2020degrees}, which corresponds to $\frac{2c}{\pi}$ in our normalization. 
\end{remark}

\subsection{Eigenfunction Evaluation}
\label{sec:eigenfunction_computation}
The eigenfunctions $\{\phi_n(x)\}_{n=0}^{\infty}$ are the well-known PSWFs \cite{Slepian1961Prolate}, which form a complete basis in $L^2[-1, 1]$, the space of square-integrable functions on $[-1, 1]$. The evaluation of PSWFs is based on the self-adjoint operator $\mathcal{Q}_c$, defined as:
\begin{equation}
	\mathcal{Q}_c[\phi](x) \triangleq \int_{-1}^{1}\frac{\sin(c(x-x'))}{\pi(x-x')}\phi(x')\,dx'.
	\label{eq:operator_Q}
\end{equation} 
It can be shown that
$$
	\mathcal{Q}_c = \frac{c}{2\pi} F_c^* F_c,
$$
where $F_c^*$ denotes the adjoint operator of $F_c$ given in (\ref{eq:PSWF_definition}).
Consequently, the same PSWFs in \eqref{eq:PSWF_definition} satisfy the following self-adjoint concentration eigenvalue problem.
$$
	\int_{-1}^{1}\frac{\sin(c(x-x'))}{\pi(x-x')}\phi_n(x')\,dx' = \mu_n \phi_n(x), \quad
	n = 0, 1, 2, \ldots,
	\label{eq:PSWF_integral}
$$
where the real-valued eigenvalues $\mu_n$ are given by
\begin{equation}
	\mu_n = \frac{c}{2\pi}|\lambda_n|^2.
	\label{eq:mu_n_definition}
\end{equation}
	
While $\phi_n(x)$ does not admit an elementary closed-form expression, the PSWFs can be evaluated efficiently by exploiting the differential operator that commutes with $\mathcal{Q}_c$~\cite{Slepian1961Prolate}. Expanding the resulting differential eigenvalue problem in orthonormal Legendre polynomials leads to the matrix eigenvalue problem:
\begin{equation}
	\mathbf{A}\boldsymbol{\beta}_n = \chi_n\boldsymbol{\beta}_n, \quad \forall n \geq 0,
	\label{eq:eigenvalue_decomposition}
\end{equation}
where 
$\mathbf{A}$ is a symmetric penta-diagonal matrix whose non-zero entries are given by:
\begin{align}
	&a_{m, m} = m(m+1) + \frac{2m(m+1) - 1}{(2m+3)(2m - 1)}c^2, \nonumber \\
	&a_{m, m+2} = a_{m+2, m} = \nonumber \\ &\quad\quad\frac{(m+1)(m+2)}{(2m+3)\sqrt{(2m+1)(2m+5)}} c^2,  \; m = 0,1,\ldots \nonumber
	\label{eq:A_entries}
\end{align}
and $\boldsymbol{\beta}_n = [\beta_{n, 0}, \beta_{n, 1}, ...]^T$ contains the linear combination coefficients for the expansion of $\phi_n(x)$ in terms of the orthonormal Legendre polynomials, $\{\bar{P}_m(x)\}$:
\begin{equation}
	\phi_n(x) = \sum_{m=0}^{\infty}\beta_{n, m}\bar{P}_m(x).
	\label{eq:Legendre_linear_combination}
\end{equation}
In \eqref{eq:eigenvalue_decomposition}, $\chi_n$ is the eigenvalue corresponding to the eigenvector $\boldsymbol{\beta}_n$ associated with the commuting differential operator. For numerical evaluation, the infinite eigensystem must be truncated to a finite dimension by choosing a cutoff number, $M$ (i.e., $\boldsymbol{\beta}_n \in \mathbb{R}^{M \times 1}$ and $\mathbf{A} \in \mathbb{R}^{M \times M}$). Guidance on the selection of an appropriate value for  this polynomial expansion truncation $M$ can be found in~\cite{BOYD2003Large, schmutzhard2015numerical}.
Appendix \ref{appdix: PSWF_calculation_derivation} provides a concise derivation of \eqref{eq:eigenvalue_decomposition} and the relevant properties of Legendre polynomials.
	
\subsection{Eigenvalue Evaluation}
A stable approach to compute the eigenvalues $\{\mu_n\}_{n=0}^{N_{1D}-1}$ is to take advantage of the parity properties of PSWFs. 
Specifically, this approach first involves computing the magnitudes of the finite-Fourier eigenvalues, $|\lambda_n|$, using separate expressions for the even and odd indices:
\begin{equation}
	|\lambda_{2p}| = (-1)^p \frac{\sqrt{2}\beta_{2p, 0}}{\phi_{2p}(0)}, \; p \geq 0,
	\label{eq:lambda_computation_even}
\end{equation}
\begin{equation}
	|\lambda_{2p+1}| = (-1)^p \sqrt{\frac{2}{3}}\frac{c\beta_{2p+1, 1}}{\partial_{x} \phi_{2p+1}(0)}, \; p\geq 0.
	\label{eq:lambda_computation_odd}
\end{equation}
The eigenvalues $\{\mu_n\}_{n=0}^{N_{1D}-1}$ can then be computed using \eqref{eq:mu_n_definition}.
A derivation is provided in Appendix \ref{appdix: PSWF_eigenvalue_computation}, and a similar approach is also reported in \cite{wang2017review}. 
The resulting PSWFs and eigenvalues provide the 1D ingredients for the tensor-product construction and spectral analysis of square apertures in the next section.
	
\section{Square Aperture}
\label{sec:two dimensional case}

In this section, we focus our analysis on the square aperture case (i.e., $L_x = L_y = L$). Building upon the 1D spectral properties established in the previous section, we address the geometric mismatch arising from a square spatial aperture and a circular wavenumber support. 
We construct a tensor-PSWF representation of the resulting non-separable concentration operator and rigorously quantify its whole-spectrum approximation error, including an explicit truncation threshold for entering the super-exponential convergence regime.
	
\subsection{Domain Mismatch}
\label{subsec:Domain_mismatch}
In the considered case, a key challenge arises from a domain mismatch between the spatial and wavenumber domains. To illustrate this, we denote the 2D small-scale fading as $h(x, y)$ by setting $z = 0$. Based on the band-limited nature of $h(x, y)$ as specified in \eqref{eq:D_kappa}, we observe that $h(x, y)$ is observed over a square spatial aperture, while its wavenumber support is confined to a circular disk.

On the other hand, the tensor products
$\{\phi_j(x)\phi_\ell(y)\}_{j,\ell=0}^{\infty}$ form a complete orthonormal basis for
$L^2([-1,1]\times[-1,1])$~\cite[Chapter II.4]{reed1980methods}.
By enumerating the index pairs $(j,\ell)$ with a single index $n$, we denote these basis functions by
$$
	\psi_n(x, y) = \phi_j(x)\phi_\ell(y), 
	\quad j, \ell = 0,1,\ldots
	\label{eq:2D_PSWF_1}
$$
This follows from the fact that $\{\phi_j(x)\}_{j=0}^{\infty}$ and $\{\phi_\ell(y)\}_{\ell=0}^{\infty}$ are both complete in $L^2([-1, 1])$ as mentioned in Section \ref{sec:eigenfunction_computation}. Consequently, any square-integrable function $f(x, y)$ defined on $(x, y) \in [-1, 1]\times[-1,1]$ can be represented by
$$
	f(x, y) = \sum_{j=0}^{\infty}\sum_{\ell=0}^{\infty}a_{j\ell }\phi_j(x)\phi_\ell(y),
	\label{eq:linear_combination_h}
$$
where $\{a_{j\ell}\}$ are the linear combination coefficients.
However, the set $\{\psi_n(x, y)\}$ does not diagonalize the square--disk concentration operator, because its separable square spectral support does not match the physical circular wavenumber support.
To show this, we first note that the 2D Fourier transform of $\{\psi_n(x, y)\}$ is separable:
\begin{equation}
	\mathcal{F}[\phi_j(x)\phi_\ell(y)](k_x, k_y) = \mathcal{F}[\phi_j(x)](k_x)\mathcal{F}[\phi_\ell(y)](k_y),
	\label{eq:Fourier_seperable}
\end{equation}
where $\mathcal{F}[\cdot](k_x)$ and $\mathcal{F}[\cdot](k_y)$ denote the 1D Fourier transform along the $x$ and $y$ axes, respectively.
Moreover, according to \cite{osipov2013prolate}, the 1D Fourier transform of a PSWF is a scaled and truncated version of itself, which is given by
\begin{equation}
	\mathcal{F}[\phi_n(x)](k_x) = 
	\left\{ \begin{array}{rcl}
		\frac{2\pi}{c \lambda_n}\phi_n(\frac{k_x}{c})
		& \text{if} \; -c \leq k_x \leq c,  \\ 
		0 \quad\quad & \text{otherwise}.
	\end{array}\right.
	\label{eq:Fourier_PSWF}
\end{equation}
Based on (\ref{eq:Fourier_seperable}) and (\ref{eq:Fourier_PSWF}), we observe that the Fourier transform of $\{\psi_n(x, y)\}$ is supported on the square domain $(k_x, k_y) \in [-c, c] \times [-c, c]$. In contrast, for the physical small-scale fading $h(x, y)$, after scaling the variables to $(x, y) \in [-1, 1]\times [-1, 1]$, the corresponding wavenumber domain is the circular disk defined by $k_x^2 + k_y^2 \leq c^2$. This domain mismatch is illustrated in Fig. \ref{Fig:Kernel_2D_PSWF_Comparison}.
\begin{figure}[htbp]
	\centerline{\includegraphics[scale=0.35]{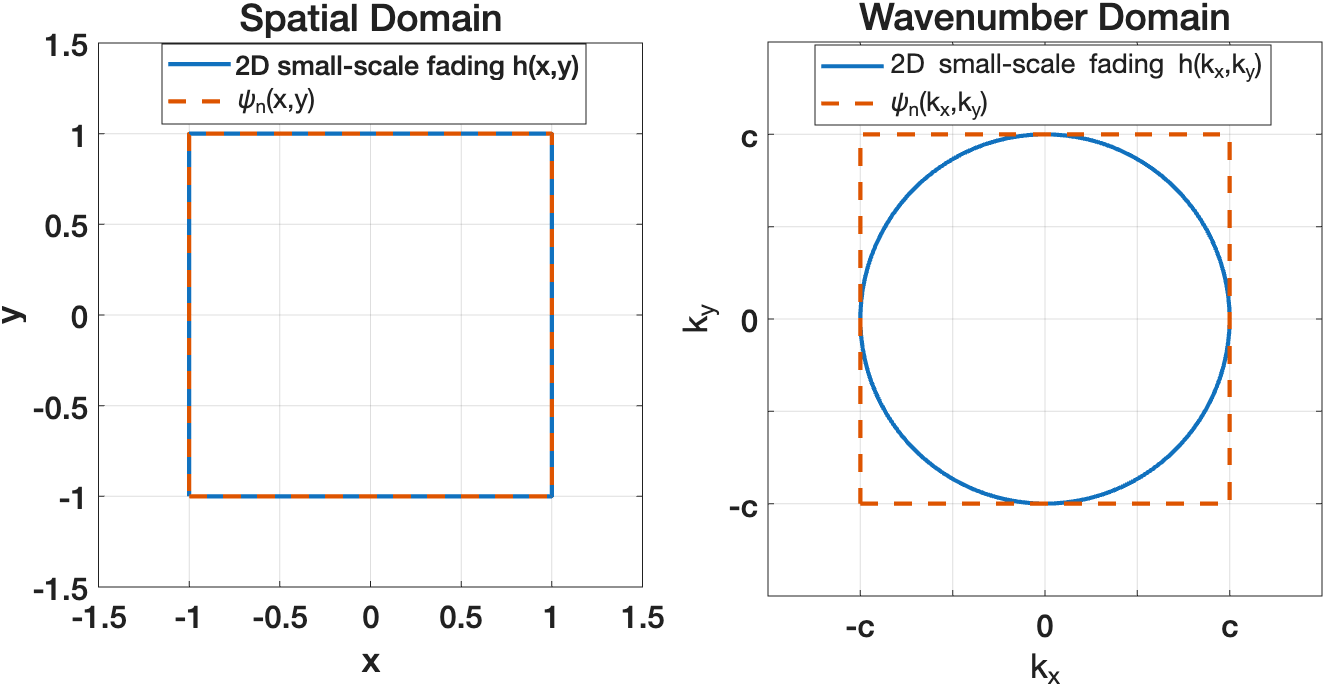}}
	\caption{2D illustration of the spatial and wavenumber domain shapes of the 2D small-scale fading, $h(x, y)$ and $\psi_n(x, y)$.}
	\label{Fig:Kernel_2D_PSWF_Comparison}
\end{figure}

\subsection{Tensor-PSWF Spectral Approximation}
\label{sec:Tensor-PSWF Spectral Approximation}
Following the same concentration principle as in the 1D case, we formulate the corresponding square--disk concentration problem and denote its orthogonal eigenfunctions by $\{\varphi_n(x, y)\}$. These functions are strictly bandlimited to the disk $k_x^2 + k_y^2 \leq c^2$ and maximize their energy concentration within the normalized square spatial domain.
The optimization problem of finding $\{\varphi_n(x, y)\}$ can be formulated as:
\begin{gather}
	\underset{\{\varphi_n(x, y)\}_{n=0}^{N-1}}{\max} \sum_{n=0}^{N-1}\mu_n \nonumber \\
	\text{s.t.} \quad \mu_n = \frac{\int_{-1}^{1}\int_{-1}^{1}|\varphi_n(x, y)|^2\,dx\,dy}{\int_{-\infty}^{\infty}\int_{-\infty}^{\infty}|\varphi_n(x, y)|^2\,dx\,dy} \nonumber \\
	\mathcal{F}[\varphi_n(x,y)](k_x, k_y) = 0 \quad \text{for} \quad k_x^2 + k_y^2 > c^2 \nonumber \\
	\int_{-\infty}^{\infty}\int_{-\infty}^{\infty}\varphi_n(x, y)\varphi_m^*(x, y)\,dx\,dy = \delta_{mn} \nonumber \\
	\text{for}\quad n,m = 0, 1, \ldots, N-1.
	\label{eq:problem_2D_optimal_PSWF} 
\end{gather}
Note that $\{\varphi_n(x, y)\}_{n=0}^{N-1}$ in \eqref{eq:problem_2D_optimal_PSWF} correspond to the leading eigenfunctions of the square-disk spatial-wavenumber concentration problem.
According to operator theory, (\ref{eq:problem_2D_optimal_PSWF}) can be reduced to finding the eigenfunctions of the spatial-wavenumber integral operator, which is formed by the composition of the spatial truncation operator, $\mathcal{P}_S$, and the wavenumber bandlimiting operator, $\mathcal{P}_c$ \cite{simons2010slepian, Slepian1964Prolate}. This leads to the eigenvalue equation:
$$
	\mu_n\varphi_n(x, y) = \mathcal{P}_c[\mathcal{P}_S[\varphi_n]](x, y).
	\label{eq:eigensystem_operator}
$$
	
In our case, the spatial truncation operator, $\mathcal{P}_S$, truncates the function outside the square $[-1, 1] \times [-1, 1]$:
\begin{equation}
	\mathcal{P}_S[\varphi](x, y) = \begin{cases} \varphi(x, y), & \left|x\right| \leq 1 \;\; \text{and} \;\; \left|y\right|  \leq 1, \\ 0, & \left|x\right|  > 1 \;\; \text{or} \;\; \left|y\right| > 1. \end{cases}
	\label{eq:spatial_operator}
\end{equation}
The wavenumber domain operator is given by
\begin{align}
	\mathcal{P}_c[\varphi](x, y) = \int_{-\infty}^{\infty}\int_{-\infty}^{\infty}G_c(x, y, x', y')\varphi(x', y')\,dx'\,dy',
	\label{eq:wavenumber_operator}
\end{align}
where
\begin{align}
	&G_c(x, y, x', y') \nonumber \\
	&\quad = \frac{1}{4\pi^2}\iint_{k_x^2 + k_y^2 \leq c^2}e^{i(k_x(x-x') + k_y(y-y'))}\,dk_x\,dk_y.
	\label{eq:low_pass_filter}
\end{align}
Given (\ref{eq:low_pass_filter}), we first note that $G_c(x, y, x', y')$ exhibits a shift-invariant property, allowing us to express it as $G_c(x, y, x', y') \equiv G_c(x-x', y-y')$. This can be interpreted as an ideal 2D low-pass filter whose wavenumber support is the disk $k_x^2 + k_y^2 \leq c^2$.

Combining (\ref{eq:spatial_operator}) with (\ref{eq:wavenumber_operator}) reduces the optimization problem (\ref{eq:problem_2D_optimal_PSWF}) to the eigenvalue problem of the non-separable continuous spatial-wavenumber integral operator, denoted as $\mathcal{P}$:
\begin{align}
	 \mu_n(\mathcal{P})\varphi_n(x, y)  &= (\mathcal{P}\varphi_n)(x, y) = \nonumber \\ &  \int_{-1}^{1}\int_{-1}^{1}G_c(x-x', y-y')\varphi_n(x', y')dx'dy', \nonumber \\
	& \qquad n = 0, 1, \ldots
	\label{eq:integral_equation}
\end{align}
	
To obtain a structured representation of~\eqref{eq:integral_equation}, we expand $\varphi_n(x,y)$ over the normalized spatial aperture in terms of the complete tensor-product basis
of 1D PSWFs:
\begin{equation}
	\varphi_n(x, y) = \sum_{j=0}^{\infty}\sum_{\ell=0}^{\infty}a_{j\ell}^n\phi_j(x)\phi_\ell(y),
	\label{eq:varphi_n}
\end{equation}
where $\{a_{j\ell}^n\}$ are the linear combination coefficients for the $n$-th eigenfunction. 
As mentioned in Section \ref{subsec:Domain_mismatch}, the tensor product basis $\{\phi_j(x)\phi_\ell(y)\}_{j,\ell=0}^{\infty}$ is complete in $L^2([-1, 1]\times [-1, 1])$~\cite[Chapter II.4]{reed1980methods}. Consequently, the infinite series expansion in \eqref{eq:varphi_n} converges to $\varphi_n(x, y)$ in $L^2([-1, 1]\times[-1, 1])$.
Substituting (\ref{eq:varphi_n}) into both sides of (\ref{eq:integral_equation}) maps the continuous integral equation into an equivalent infinite-dimensional discrete eigenvalue problem. To do this, we apply the following coefficient-extraction linear functional to both sides of (\ref{eq:integral_equation}):
$$
	\mathcal{P}_{pq}[f] = \int_{-1}^{1}\int_{-1}^{1}\phi_p(x)\phi_q(y)f(x, y)\,dx\,dy, \; p, q = 0, 1, 2, \ldots
	\label{eq:projection_operator}
$$
	
The left-hand side of (\ref{eq:integral_equation}) becomes
\begin{align}
	&\mu_n(\mathcal{P})\sum_{j=0}^{\infty}\sum_{\ell=0}^{\infty}a_{j\ell}^n\int_{-1}^{1}\int_{-1}^ {1}\phi_p(x)\phi_q(y)\phi_j(x)\phi_\ell(y)dxdy \nonumber \\
	&\quad = \mu_n(\mathcal{P})a_{pq}^n.
	\label{eq:LHS}
\end{align}
The equality in (\ref{eq:LHS}) holds because the 1D PSWFs are orthonormal in the interval $[-1, 1]$ (i.e., $\int_{-1}^{1}\phi_p(x)\phi_j(x)\,dx = \delta_{pj}$). 
	
We now consider the right-hand side of (\ref{eq:integral_equation}), which is:
$$
	\sum_{j=0}^{\infty}\sum_{\ell=0}^{\infty}M_{(pq)(j\ell)}a_{j\ell}^n,
	\label{eq:RHS}
$$
where
\begin{align}
	M_{(pq)(j\ell)} = &\iiiint_{[-1,1]^4} \phi_p(x)\phi_q(y)G_c(x-x', y-y')\nonumber \\ 
	&\qquad\qquad \times \phi_j(x')\phi_\ell(y') \, dx' \, dy' \, dx \, dy.
	\label{eq:M_pq_jl}
\end{align}
Combining this result with the left-hand side, we obtain the infinite-dimensional discrete eigenvalue problem:
\begin{equation}
	\mathbf{M}\mathbf{a}_n = \mu_n(\mathbf{M}) \mathbf{a}_n,
	\label{eq:linear_algebra}
\end{equation}
where $\mathbf{a}_n$ denotes the $n$-th eigenvector. 
Because the tensor-product PSWFs form a complete orthonormal basis of $L^2([-1, 1]\times[-1, 1])$, the infinite matrix $\mathbf{M}$ provides an exact matrix representation of the continuous operator $\mathcal{P}$ in this basis.
As established in functional analysis~\cite{hackbusch2011integral, reed1980methods}, $\mathbf{M}$ and $\mathcal{P}$ therefore share the identical set of non-zero eigenvalues $\mu_n$:
\begin{equation}
	\mu_n(\mathbf{M}) = \mu_n(\mathcal{P}), \quad n = 0, 1, 2, \ldots
	\label{eq:eigenvalue_equivalence}
\end{equation}
	
We now truncate the infinite-dimensional problem in (\ref{eq:linear_algebra}) into a finite-dimensional version:
\begin{equation}
	\mathbf{M}_N\tilde{\mathbf{a}}_n = \mu_n(\mathbf{M}_N) \tilde{\mathbf{a}}_n, \quad n = 0, \ldots, N-1,
	\label{eq:EVD_N}
\end{equation}
where $\mathbf{M}_N \in \mathbb{R}^{N \times N}$ with $N = N_xN_y$. Here, $N_x$ and $N_y$ denote the number of 1D PSWFs retained along the $x$ and $y$ axes, respectively.
For our considered square aperture, we have $N_x = N_y \triangleq N_{1D}$, which yields $N = N_{1D}^2$.
By mapping the 2D indices to a 1D sequence, the $(pN_{1D} + q+1, jN_{1D} + \ell + 1)$-th element of $\mathbf{M}_N$ is denoted as $M_{(pq)(j\ell)}$. The $n$-th eigenvector of $\mathbf{M}_N$ is denoted as $\tilde{\mathbf{a}}_n \in \mathbb{R}^{N \times 1}$, whose $(jN_{1D} + \ell + 1)$-th element is $\tilde{a}_{j\ell}^n$.
Consequently, for $0\le n<N$, the $n$-th eigenfunction for the continuous operator $\mathcal{P}$ in \eqref{eq:integral_equation} is approximated as:
$$
	\tilde{\varphi}_n(x, y) = \sum_{j=0}^{N_{1D} - 1}\sum_{\ell=0}^{N_{1D} - 1}  \tilde{a}_{j\ell}^n\phi_j(x)\phi_\ell(y).
	\label{eq:varphi_n_N}
$$
	
The calculation of the matrix elements $M_{(pq)(j\ell)}$ in (\ref{eq:M_pq_jl}) initially incurs a high-complexity four-dimensional integral. However, by exploiting the Fourier transform and the parity properties of the 1D PSWFs, we can rigorously reduce its dimensionality and establish a sparse structure, as summarized in the following lemma.
\begin{lemma}
	\label{lemma:M_pqjl_calculation}
	For any given indices $p, q, j, \ell \geq 0$, the four-dimensional integral $M_{(pq)(j\ell)}$ in \eqref{eq:M_pq_jl} satisfies the following properties:
	
	If $\operatorname{mod}(p+j, 2) = 1$ or $\operatorname{mod}(q+\ell, 2) = 1$, the element strictly vanishes:
	$$
		M_{(pq)(j\ell)} = 0. \label{eq:odd_function_1}
	$$
	Otherwise, i.e., if $\operatorname{mod}(p+j, 2) = 0$ and $\operatorname{mod}(q+\ell, 2) = 0$, the integral is reduced to a two-dimensional polar form in the wavenumber domain:
	\begin{align}
		M_{(pq)(j\ell)} = & \pm\frac{\sqrt{\mu_p\mu_q\mu_j\mu_\ell}}{c^2}\int_{0}^{2\pi}\int_{0}^{c} \phi_p\Big(\frac{r\cos\theta}{c}\Big) \nonumber \\ 
		& \times \phi_q\Big(\frac{r\sin\theta}{c}\Big)\phi_j\Big(\frac{r\cos\theta}{c}\Big)\phi_\ell\Big(\frac{r\sin\theta}{c}\Big)r \, dr \, d\theta.
		\label{eq:even_function_2}
	\end{align}
	The sign in (\ref{eq:even_function_2}) is positive if $\operatorname{mod}(p+q-j-\ell, 4) = 0$ and negative if $\operatorname{mod}(p+q-j-\ell, 4) = 2$. 
\end{lemma}
A detailed derivation of Lemma~\ref{lemma:M_pqjl_calculation} is provided in Appendix~\ref{appdix:2D_PSWF_Calculation}. 
\begin{remark}
	Under the conditions $\operatorname{mod}(p + j, 2) = 0$ and $\operatorname{mod}(q+\ell, 2) = 0$, $p-j$ and $q-\ell$ are both even, so that $\operatorname{mod}(p+q-j-\ell, 4)$ can only be $0$ or $2$.
\end{remark}
\begin{remark}
	Lemma~\ref{lemma:M_pqjl_calculation} indicates that $\mathbf{M}$ is a highly sparse matrix. As its truncated version, $\mathbf{M}_N$ inherits this sparsity. From a computational efficiency perspective, this property drastically reduces the number of non-zero elements to be computed. Moreover, among the remaining non-zero elements, the required calculations are further minimized: $\mathbf{M}_N$ is symmetric (requiring only the upper triangular part), and the index symmetries in \eqref{eq:even_function_2} lead to identical integral evaluations for permuted index combinations, thereby further avoiding redundant computations.
\end{remark}
	
Given the structure established in Lemma~\ref{lemma:M_pqjl_calculation}, the finite matrix $\mathbf{M}_N$ is the matrix representation obtained by projecting $\mathcal{P}$ onto the tensor-PSWF subspace spanned by the first $N_{1D}$ one-dimensional basis functions along each coordinate. This projection structure allows the spectral approximation error to be controlled directly through a trace defect, leading to the following result.
\begin{theorem}
	\label{theorem:approx_error}
	Let $\mu_n(\mathcal{P})$ be the $n$-th eigenvalue of the
	non-separable continuous spatial-wavenumber integral operator
	$\mathcal{P}$, arranged in non-increasing order, and let
	$\mu_n(\mathbf{M}_N)$ be the $n$-th eigenvalue of the
	finite-dimensional truncated matrix $\mathbf{M}_N$, also arranged
	in non-increasing order, with the symmetric truncation
	$N=N_{1D}^2$. If
	$$
		N_{1D}>\frac{ec}{4},
	$$
	then the total absolute spectral approximation error
	\begin{align}
		&\sum_{n=0}^{N-1}
		\left|
		\mu_n(\mathcal{P})-\mu_n(\mathbf{M}_N)
		\right|
		+
		\sum_{n=N}^{\infty}\mu_n(\mathcal{P})
		\nonumber\\
		= & 
		\operatorname{Tr}(\mathcal{P})
		-
		\operatorname{Tr}(\mathbf{M}_N) <
		C_1(N_{1D},c)
		e^{-2N_{1D}
			\ln\left(\frac{4N_{1D}}{ce}\right)}, \nonumber
		\label{eq:error_truncation}
	\end{align}
	where
	\begin{equation}
		C_1(N_{1D},c)
		=
		\frac{c^2}{\pi N_{1D}}
		\left[
		1-
		\left(
		\frac{ec}{4(N_{1D}+1)}
		\right)^2
		\right]^{-1}.
		\label{eq:C1_definition}
	\end{equation}
\end{theorem}
\begin{proof}
	See Appendix~\ref{appendix:truncation_error}.
\end{proof}
\begin{remark}
	From Theorem~\ref{theorem:approx_error}, we see that once $N_{1D}>ec/4$, the total absolute spectral error between the zero-padded
	eigenspectrum of $\mathbf{M}_N$ and that of $\mathcal{P}$ admits an explicit non-asymptotic
	super-exponential upper envelope. 
\end{remark}

It is worth emphasizing that, although the square--disk mismatch destroys separability and induces non-diagonal coupling in the tensor-PSWF representation, the whole-spectrum approximation error remains controlled by the tail of the underlying 1D PSWF eigenvalues. In particular, since these 1D eigenvalues satisfy the explicit envelope $\mu_n(c)
< \frac{c}{4n}
e^{-2n\ln\left(\frac{4n}{ce}\right)}$ as derived in~\eqref{eq:inequality_n_c}, the resulting whole-spectrum error retains a certified super-exponential decay beyond the 1D truncation threshold. 
In fact, the trace-defect mechanism underlying Theorem~\ref{theorem:approx_error} also extends to rectangular apertures, as summarized next.
\begin{remark}
	For a rectangular aperture of side lengths $L_x$ and $L_y$, let $c_x = \frac{\kappa L_x}{2}$ and $c_y = \frac{\kappa L_y}{2}$. We further let $\mathcal{P}_{c_x, c_y}$ denote the corresponding normalized concentration operator, and let $\mathbf{M}_{N_x, N_y}$ denote its matrix representation obtained by projecting $\mathcal{P}_{c_x, c_y}$ onto the tensor-PSWF subspace spanned by the first $N_x$ and $N_y$ 1D PSWFs with parameters $c_x$ and $c_y$, respectively. Define
	$$ T_x = \sum_{p=N_x}^{\infty}\mu_p(c_x),\quad T_y = \sum_{q=N_y}^{\infty}\mu_q(c_y).$$
	Using the derivation approach in Appendix~\ref{appendix:truncation_error}, the trace defect for the rectangular apertures satisfies
	\begin{align}
		&\operatorname{Tr}(\mathcal{P}_{c_x, c_y}) - \operatorname{Tr}(\mathbf{M}_{N_x, N_y}) \nonumber \\ & \qquad \leq \frac{2c_y}{\pi}T_x + \frac{2c_x}{\pi}T_y - T_xT_y \leq \frac{2c_y}{\pi}T_x + \frac{2c_x}{\pi}T_y. \nonumber
	\end{align}
	Hence, the rectangular-aperture approximation error remains governed by the corresponding 1D PSWF eigenvalue tails, with the square result recovered by setting $c_x = c_y = c$ and $N_x = N_y = N_{1D}$.
\end{remark}

For the square-aperture setting, a direct bound on every eigenvalue in the flattened two-dimensional ordering can also be obtained rather than only at perfect-square truncation indices.
\begin{corollary}
	\label{corollary:eigenspectrum_decay}
	Define $m_n \triangleq \left\lfloor\sqrt{n}\right\rfloor.$
	If $m_n>ec/4$, then the $n$-th eigenvalue of the continuous
	operator $\mathcal{P}$ satisfies
	\begin{align}
		\mu_n(\mathcal{P})
		<
		\frac{C_1(m_n,c)}
		{n-m_n^2+1}
		e^{-2m_n
			\ln\left(\frac{4m_n}{ce}\right)}.
		\label{eq:UB_eigenvalue}
	\end{align}
	Consequently, the flattened two-dimensional eigenspectrum
	admits the asymptotic upper envelope
	\begin{equation}
		\ln\mu_n(\mathcal{P})
		\leq
		-2\sqrt{n}
		\ln\left(\frac{4\sqrt{n}}{ce}\right)
		+\mathcal{O}(\ln n).
		\label{eq:flattened_asymptotic}
	\end{equation}
\end{corollary}
\begin{proof}
	Since $m_n^2\leq n<(m_n+1)^2$ and the eigenvalues of
	$\mathcal{P}$ are arranged in non-increasing order,
	$$
		(n-m_n^2+1)\mu_n(\mathcal{P})
		\leq
		\sum_{k=m_n^2}^{n}\mu_k(\mathcal{P})\leq
		\sum_{k=m_n^2}^{\infty}\mu_k(\mathcal{P}).
	$$
	Applying Theorem~\ref{theorem:approx_error} with the
	one-dimensional truncation dimension set to $m_n$ yields
	\eqref{eq:UB_eigenvalue}. Since
	$m_n=\sqrt{n}+\mathcal{O}(1)$, taking the logarithm of
	\eqref{eq:UB_eigenvalue} further gives
	\eqref{eq:flattened_asymptotic}.
\end{proof}

\begin{remark}
	\label{remark:perturbation_tail_error}
	It is worth distinguishing the 1D truncation threshold
	$N_{1D}>ec/4$ from the classical 1D spatial DoF $2c/\pi$. While $2c/\pi$ equals
	the trace of the 1D concentration operator $\mathcal{Q}_c$
	and characterizes its macroscopic effective dimension, the
	threshold $ec/4$ marks the analytical onset at which the
	explicit PSWF tail envelope becomes super-exponentially
	decaying.
\end{remark}

\begin{remark}
	\label{remark:decaying_rate}
	Corollary~\ref{corollary:eigenspectrum_decay} reveals the
	dimensionality-folding behavior of the two-dimensional
	eigenspectrum. In contrast to the one-dimensional PSWF
	eigenvalues, whose explicit envelope decays as
	$e^{-2p\ln(4p/(ce))}$ with respect to the one-dimensional
	index $p$, the flattened two-dimensional eigenvalue index
	introduces the scaling $p\sim\sqrt{n}$. Consequently, the
	dominant term in \eqref{eq:flattened_asymptotic} becomes
	$$
		e^{-2\sqrt{n}
			\ln\left(\frac{4\sqrt{n}}{ce}\right)}.
	$$
	This result establishes an explicit upper envelope rather than
	an exact asymptotic equivalence; the actual eigenspectrum may
	decay faster than the derived bound.
\end{remark}

\section{Channel Capacity}
\label{sec:channel_capacity}
In this section, we quantify the information-theoretic loss incurred
when the spatially continuous channel is restricted to the finite
tensor-PSWF subspace developed in Section~\ref{sec:two dimensional case}. In contrast to an
analytical upper bound obtained by applying Jensen's inequality, our
analysis directly considers the actual ergodic capacity, in which the
instantaneous mutual information is averaged over the random channel
and the transmit covariance is optimized subject to the total power
constraint. We first establish the equivalent spatial-eigenmode
representation and actual ergodic capacity of the continuous channel,
and then compare it with the corresponding finite-dimensional PSWF-truncated
channel.

\subsection{Ergodic Capacity of the Continuous Aperture Channel}
\label{subsec:capacity_limit}

We start from the continuous electromagnetic signal model between
the transmitter and receiver. For two-dimensional planar apertures,
the spatially continuous input-output relation is given by
~\cite{Poon2005Degrees,Poon2006Impact}
\begin{equation}
	Y(\mathbf{r}_r)
	=
	\int_{\mathcal{D}_t}
	h(\mathbf{r}_r,\mathbf{r}_t)
	X(\mathbf{r}_t)
	d\mathbf{r}_t
	+
	Z(\mathbf{r}_r),
	\qquad
	\mathbf{r}_r\in\mathcal{D}_r,
	\label{eq:signal_model}
\end{equation}
where $X\in L^2(\mathcal{D}_t)$ and
$Y\in L^2(\mathcal{D}_r)$ denote the transmitted and received
fields, respectively. The additive noise $Z$ is a zero-mean
circularly symmetric complex Gaussian spatial white-noise field
satisfying $\mathbb{E}[Z(\mathbf{r}_r)Z^*(\mathbf{r}'_r)]
=\sigma_z^2\delta(\mathbf{r}_r-\mathbf{r}'_r)$.
The channel kernel $h\in L^2(\mathcal{D}_r\times\mathcal{D}_t)$ maps
the transmit aperture $\mathcal{D}_t\subset\mathbb{R}^2$ to the
receive aperture $\mathcal{D}_r\subset\mathbb{R}^2$. We assume that \(h(\mathbf r_r,\mathbf r_t)\) is a zero-mean proper complex Gaussian random field.

For the capacity analysis, we assume perfect instantaneous CSIR and statistical CSIT. We consider a symmetric transmission scenario with identical square transmit and receive
apertures, i.e., $\mathcal{D}_t=\mathcal{D}_r$. In the following,
the spatial coordinates are understood under the same normalization
used in Section~\ref{sec:two dimensional case}, such that both apertures are mapped to
$[-1,1]\times[-1,1]$ and the corresponding circular wavenumber
support is given by $k_x^2+k_y^2\leq c^2$.

Under the separable scattering assumption, the continuous channel correlation function is modeled as
\begin{align}
	\mathbb{E}\!\left[
	h(\mathbf{r}_r,\mathbf{r}_t)
	h^*(\mathbf{r}'_r,\mathbf{r}'_t)
	\right] =
	R_r(\mathbf{r}_r,\mathbf{r}'_r)
	R_t^*(\mathbf{r}_t,\mathbf{r}'_t),
	\label{eq:separable_correlation}
\end{align}
where $R_r(\mathbf{r}_r,\mathbf{r}'_r)$ and
$R_t(\mathbf{r}_t,\mathbf{r}'_t)$ denote the receive and transmit spatial autocorrelation functions, respectively. For the symmetric
square--disk channel considered in this paper, we further specify
these autocorrelation kernels, up to constant power-scaling factors,
by the circularly bandlimited kernel $G_c$ introduced in
\eqref{eq:low_pass_filter}:
\begin{align}
	R_r(\mathbf{r}_r,\mathbf{r}'_r)
	&=
	\gamma_r G_c(\mathbf{r}_r-\mathbf{r}'_r),
	\label{eq:Rr_Gc}\\
	R_t(\mathbf{r}_t,\mathbf{r}'_t)
	&=
	\gamma_t G_c(\mathbf{r}_t-\mathbf{r}'_t),
	\label{eq:Rt_Gc}
\end{align}
where $\gamma_r>0$ and $\gamma_t>0$ denote the receive and
transmit scattering-power factors, respectively. Comparing
\eqref{eq:Rr_Gc} and \eqref{eq:Rt_Gc} with the kernel of the
continuous integral operator $\mathcal{P}$ in
\eqref{eq:integral_equation}, the corresponding receive and transmit
covariance eigenfunctions coincide with the eigenfunctions of
$\mathcal{P}$. Their eigenvalues therefore satisfy
$\mu_{r,n}=\gamma_r\mu_n(\mathcal{P})$ and
$\mu_{t,n}=\gamma_t\mu_n(\mathcal{P})$. Since the product
$\gamma_r\gamma_t$ only scales the overall channel power, it is
absorbed into the effective signal-to-noise ratio (SNR) in the subsequent analysis.
Accordingly, without loss of generality for the normalized channel
model, we set $\gamma_r=\gamma_t=1$, so that
\begin{equation}
	\mu_{r,n}
	=
	\mu_{t,n}
	=
	\mu_n(\mathcal{P}),
	\qquad n=0,1,\ldots
	\label{eq:symmetric_covariance_eigenvalues}
\end{equation}

By Mercer's theorem and the Karhunen-Lo\`eve expansion of $h(\mathbf{r}_r, \mathbf{r}_t)$, the continuous input-output relation given in \eqref{eq:signal_model} can be converted to the following discrete form:
\begin{equation}
	\mathbf{y}
	=
	\mathbf{H}\mathbf{x}
	+
	\mathbf{z},
	\label{eq:equivalent_channel}
\end{equation}
where $\mathbf{x}=[x_0,x_1,\ldots]^T$,
$\mathbf{y}=[y_0,y_1,\ldots]^T$, and
$\mathbf{z}=[z_0,z_1,\ldots]^T$. $\mathbf{H}$ can be compactly represented as
\begin{equation}
	\mathbf{H}
	=
	\boldsymbol{\Lambda}^{1/2}
	\mathbf{H}_w
	\boldsymbol{\Lambda}^{1/2},
	\label{eq:H_factored}
\end{equation}
where
\begin{equation}
	\boldsymbol{\Lambda}
	\triangleq
	\operatorname{diag}
	\left(
	\mu_0(\mathcal{P}),
	\mu_1(\mathcal{P}),
	\ldots
	\right),
	\label{eq:Lambda_definition}
\end{equation}
and the entries $\{w_{n,m}\}$ of $\mathbf{H}_w$ are i.i.d. $\mathcal{CN}(0, 1)$ random variables. Consequently, the $(n,m)$-th entry of $\mathbf{H}$ is given as:
$$
	[\mathbf{H}]_{n,m}
	=
	\sqrt{\mu_n(\mathcal{P})}\,
	w_{n,m}\,
	\sqrt{\mu_m(\mathcal{P})}.
	\label{eq:H_element}
$$
We let $P_T>0$ denote the total transmit-power and the noise components $z_n$ are i.i.d $\mathcal{CN}(0, \sigma_z^2)$ due to the orthonormal projection. The SNR is defined as $\rho\triangleq P_T/\sigma_z^2$. Consequently, we have the following ergodic capacity result for the infinite dimensional channel $\mathbf{H}$ in the following lemma.

\begin{lemma}
	\label{lemma:continuous_capacity}
	Let 
	\begin{equation}
		\mathcal{I}(\mathbf{Q}) = \mathbb{E}\left[\log_2\det\left(
		\mathbf{I}
		+
		\frac{1}{\sigma_z^2}
		\mathbf{H}\mathbf{Q}\mathbf{H}^{\dagger}
		\right)
		\right],
		\label{eq:I_Q_Definition}
	\end{equation}
	where $\mathbf{Q} \triangleq \mathbb{E}[\mathbf{x}\mathbf{x}^{\dagger}]$ denotes the deterministic transmit covariance matrix and the determinant is understood in the Fredholm sense. Under perfect instantaneous CSIR and statistical CSIT, 
	the actual ergodic capacity of the continuous
	symmetric square--disk channel in~\eqref{eq:equivalent_channel} is
	\begin{align}
		C_{\mathrm{erg}}
		=
		\sup_{\substack{
				\mathbf{Q}\succeq0\\
				\operatorname{Tr}(\mathbf{Q})\leq P_T
		}}
		\mathcal{I}(\mathbf{Q}).
		\label{eq:actual_ergodic_capacity}
	\end{align}
	Further, the capacity optimization can be restricted to diagonal covariance
	matrices without loss of optimality. That is
	\begin{align}
		C_{\mathrm{erg}}
		=
		\sup_{\substack{
				\widetilde{\mathbf{Q}}\succeq0,\;
				\widetilde{\mathbf{Q}}\ {\rm diagonal}\\
				\operatorname{Tr}(\widetilde{\mathbf{Q}})\leq P_T
		}}
		\mathcal{I}(\widetilde{\mathbf{Q}}).
		\label{eq:ergodic_capacity_diagonal_Q}
	\end{align}
\end{lemma}
The derivation of the discrete representation in~\eqref{eq:equivalent_channel} and the proof of Lemma~\ref{lemma:continuous_capacity} is given in Appendix~\ref{appendix:channel_representations}.

We note that \eqref{eq:actual_ergodic_capacity} is the actual ergodic capacity considered in the remainder of this paper rather than a capacity upper bound based on Jensen's inequality. Importantly, \eqref{eq:ergodic_capacity_diagonal_Q} does
not require the optimal power allocation $\widetilde{\mathbf{Q}}$ to
be determined explicitly; only the non-negativity and total-power
constraint will be required in the subsequent convergence analysis.

\subsection{PSWF-Truncated Random Channel}
We now construct the finite-dimensional random channel associated
with the tensor-PSWF truncation developed in
Section~\ref{sec:two dimensional case}. Recall that the retained
spatial subspace is spanned by the tensor-product PSWFs
$\{\phi_p(x)\phi_q(y)\}_{0\leq p,q<N_{1D}}$, with
$N=N_{1D}^2$. 

Under the same separable Gaussian scattering model used for the
continuous channel, we first propose that the projected channel can be represented in the eigenbasis of $\mathbf{M}_N$, without loss of distribution, as 
\begin{equation}
	\mathbf{H}_N
	=
	\boldsymbol{\Lambda}_N^{1/2}
	\mathbf{H}_{w,N}
	\boldsymbol{\Lambda}_N^{1/2},
	\label{eq:HN_eigenbasis}
\end{equation}
where 
\begin{equation}
	\boldsymbol{\Lambda}_N
	\triangleq
	\operatorname{diag}
	\left(
	\mu_0(\mathbf{M}_N),
	\mu_1(\mathbf{M}_N),
	\ldots,
	\mu_{N-1}(\mathbf{M}_N)
	\right),
	\label{eq:Lambda_N_definition}
\end{equation}
and $\mathbf{H}_{w,N}\in\mathbb{C}^{N\times N}$ has i.i.d. $\mathcal{CN}(0,1)$ entries. We now have the following ergodic capacity result for the projected channel.
\begin{lemma}
	\label{lemma:truncated_capacity}
	Let $$\mathcal{I}_N(\mathbf{Q}_N) = \mathbb{E}
	\left[
	\log_2
	\det
	\left(
	\mathbf{I}_N
	+
	\frac{1}{\sigma_z^2}
	\mathbf{H}_N
	\mathbf{Q}_N
	\mathbf{H}_N^\dagger
	\right)
	\right], $$
	where $\mathbf{Q}_N$ denotes the truncated transmit covariance matrix.
	For the tensor-PSWF projected channel with $N=N_{1D}^2$, 
	the actual ergodic capacity is
	\begin{align}
		C_{\mathrm{erg}}^{(N)}
		\triangleq
		\max_{\substack{
				\mathbf{Q}_N\succeq0\\
				\operatorname{Tr}(\mathbf{Q}_N)\leq P_T
		}}
		\mathcal{I}_N(\mathbf{Q}_N).
		\label{eq:finite_ergodic_capacity}
	\end{align}
	Moreover, the maximization can be restricted, without loss of
	optimality, to transmit covariance matrices diagonal in the
	eigenbasis of $\mathbf{M}_N$. That is
	\begin{align}
		C_{\mathrm{erg}}^{(N)}
		= 
		\max_{\substack{
					\widetilde{\mathbf{Q}}_N\succeq0,\;
					\widetilde{\mathbf{Q}}_N\ {\rm diagonal}\\
					\operatorname{Tr}(\widetilde{\mathbf{Q}}_N)\leq P_T
		}}
		\mathcal{I}_N(\widetilde{\mathbf{Q}}_N).
		\label{eq:finite_capacity_diagonal_Q}
	\end{align}
\end{lemma}
The derivation for the projected channel representation in~\eqref{eq:HN_eigenbasis} and the proof of Lemma~\ref{lemma:truncated_capacity} is given in Appendix~\ref{appendix:channel_representations}.

Hence, $\mathbf{H}_N$ is not an independently introduced
finite-dimensional surrogate, but the Gaussian random channel induced
by projecting the continuous square--disk channel onto the same
tensor-PSWF subspace used to construct $\mathbf{M}_N$. Consequently,
comparing $C_{\mathrm{erg}}^{(N)}$ with $C_{\mathrm{erg}}$ directly
quantifies the information-rate loss caused by the spatial truncation
studied in Section~\ref{sec:two dimensional case}.

\subsection{Ergodic-Capacity Gap and Convergence}
Having established the ergodic capacities of the continuous
channel and the tensor-PSWF truncated channel in the preceding
subsections, we now quantify the capacity loss induced by the
finite-dimensional spatial truncation. The key connection to the
spectral approximation developed in Section~\ref{sec:two dimensional case} is the trace defect
$\operatorname{Tr}(\mathcal{P})-\operatorname{Tr}(\mathbf{M}_N)$,
which exactly characterizes the total spectral approximation error
in Theorem~\ref{theorem:approx_error}. As shown below, the same quantity directly controls
the gap between the two ergodic capacities.

\begin{theorem}
	\label{theorem:capacity_convergence}
	Let $C_{\mathrm{erg}}$ denote the ergodic capacity of the
	continuous symmetric square--disk channel, and let
	$C_{\mathrm{erg}}^{(N)}$ denote the ergodic capacity of the
	corresponding tensor-PSWF truncated channel with
	$N=N_{1D}^2$. Under perfect instantaneous CSIR and
	statistical CSIT, the capacity loss induced by the spatial
	truncation satisfies
	\begin{align}
		&0
		\leq
		C_{\mathrm{erg}}-C_{\mathrm{erg}}^{(N)}\leq\nonumber \\
		&\qquad\frac{\rho}{\ln 2}
		\left[
		\mu_0(\mathcal{P})
		+
		\operatorname{Tr}(\mathbf{M}_N)
		\right]
		\left[
		\operatorname{Tr}(\mathcal{P})
		-
		\operatorname{Tr}(\mathbf{M}_N)
		\right].
		\label{eq:capacity_trace_defect_tight}
	\end{align}
\end{theorem}
\begin{proof}
	See Appendix~\ref{appendix:capacity_convergence}.
\end{proof}
\begin{remark}
	Since
	$\mu_0(\mathcal{P})\leq1$ and
	$\operatorname{Tr}(\mathbf{M}_N)
	\leq\operatorname{Tr}(\mathcal{P})=c^2/\pi$,
	\eqref{eq:capacity_trace_defect_tight} further yields
	\begin{align}
		0
		\leq
		C_{\mathrm{erg}}-C_{\mathrm{erg}}^{(N)}
		&\leq
		\frac{\rho}{\ln 2}
		\left(
		1+\frac{c^2}{\pi}
		\right)
		\left[
		\operatorname{Tr}(\mathcal{P})
		-
		\operatorname{Tr}(\mathbf{M}_N)
		\right].
		\label{eq:capacity_trace_defect}
	\end{align}
	Moreover, when $N_{1D} > ec/4$, as required in Theorem~\ref{theorem:approx_error}, the capacity gap admits the explicit non-asymptotic upper bound
	\begin{align}
		C_{\mathrm{erg}}-C_{\mathrm{erg}}^{(N)}
		&<
		\frac{\rho}{\ln 2}
		\left(
		1+\frac{c^2}{\pi}
		\right)
		C_1(N_{1D},c)
		e^{-2N_{1D}
			\ln\left(\frac{4N_{1D}}{ce}\right)}.
		\label{eq:capacity_convergence}
	\end{align}
\end{remark}
Theorem~\ref{theorem:capacity_convergence} shows that the ergodic-capacity loss is controlled
directly by the same trace defect governing the spectral
approximation in Theorem~\ref{theorem:approx_error}.
Combining this capacity-gap bound with Theorem~\ref{theorem:approx_error} yields, once the 1D truncation threshold is exceeded, the explicit super-exponential convergence rate in~\eqref{eq:capacity_convergence}. 
\begin{remark}
	Since the proof of Theorem~\ref{theorem:capacity_convergence} depends on the spectral ordering and trace-defect structure rather than on the square geometry itself, the capacity-gap upper bound in~\eqref{eq:capacity_trace_defect_tight} also applies to the rectangular extension with $\mathcal{P}$ and  $\mathbf{M}_N$ replaced by $\mathcal{P}_{c_x, c_y}$ and $\mathbf{M}_{N_x, N_y}$, respectively.
\end{remark}

\section{Quadrature Rules and Numerical Results}
\label{sec:Quadrature Rules and Numerical Results}
\subsection{Quadrature Rules}
We now consider the quadrature error incurred when computing $M_{(pq)(j\ell)}$ in (\ref{eq:even_function_2}), which involves a two-dimensional integral in polar coordinates. To evaluate $M_{(pq)(j\ell)}$ with high accuracy and efficiency, we establish the following propositions:
\begin{proposition}
	\label{prop:radial_quadrature}
	When evaluating the inner radial integral of $M_{(pq)(j\ell)}$ in \eqref{eq:even_function_2} using Gauss-Legendre quadrature (GLQ), once the number of quadrature nodes $M_r$ exceeds $\frac{e\sqrt{2}c}{4}$, the quadrature error exhibits super-exponential convergence.
\end{proposition}
\begin{proposition}
	\label{prop:angular_quadrature}
	When evaluating the outer angular integral of $M_{(pq)(j\ell)}$ in \eqref{eq:even_function_2} using the uniform trapezoidal rule, once the number of quadrature nodes $M_{\theta}$ exceeds $\sqrt{2}ec$, the quadrature error exhibits exponential convergence.
\end{proposition}
Detailed proofs for Propositions \ref{prop:radial_quadrature} and \ref{prop:angular_quadrature} are provided in Appendix~\ref{appendix:quadrature_nodes}.
\begin{remark}
	Leveraging the super-exponential and exponential convergence properties of the radial and angular integrals, respectively, we can choose $M_r$ and $M_\theta$ as
	$$
		M_r = \left\lfloor \frac{e\sqrt{2}c}{4}\right\rfloor + \tilde{M}_r, \quad M_\theta = \left\lfloor \sqrt{2}ec\right\rfloor  + \tilde{M}_\theta,
	$$
	where $\tilde{M}_r$ and $\tilde{M}_\theta$ are positive integer offsets beyond the corresponding analytical thresholds.
\end{remark}

\subsection{Numerical Results}
\begin{figure}[htbp]
	\centerline{\includegraphics[scale=0.38]{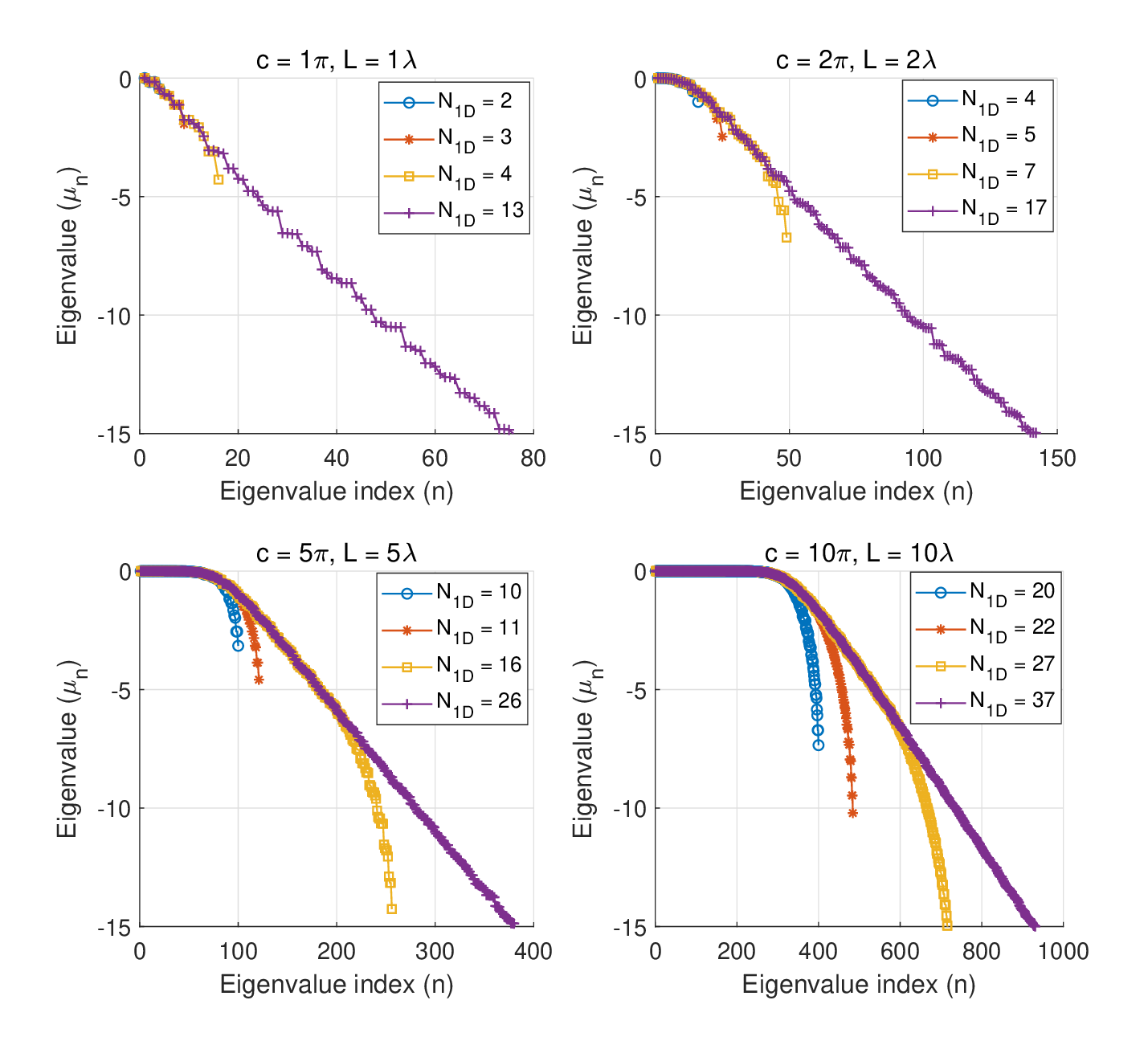}}
	\caption{Eigenvalue distribution under various settings of normalized aperture size and 1D truncation dimension $N_{1D}$.}
	\label{Fig:Eigenvalues_1D_2D_PSWFs}
\end{figure}
Figure \ref{Fig:Eigenvalues_1D_2D_PSWFs} illustrates the ordered eigenvalue distributions $\mu_n(\mathbf{M}_N)$ of the truncated sparse matrix under varying antenna apertures ($L = 1\lambda, 2\lambda, 5\lambda, 10\lambda$). For each configuration, we evaluate the discrete eigenvalue sequences reconstructed using different numbers of 1D PSWF basis functions, denoted by $N_{1D}$. Specifically, we benchmark the conventional 1D spatial DoF ($N_{1D} = 2L/\lambda$) against the derived 1D truncation threshold ($N_{1D} = \lfloor ec/4 \rfloor + 1$), alongside sufficiently large basis sets used as numerically converged references.

For each antenna aperture configuration in Fig.~\ref{Fig:Eigenvalues_1D_2D_PSWFs}, the first and second values of $N_{1D}$ are set to the conventional DoF $2L/\lambda$ and the derived threshold $\lfloor ec/4\rfloor + 1$. We observe that their numerical eigenvalues exhibit an abrupt, artificial drop, which is caused by the insufficient dimension of the subspace projection. As $N_{1D}$ continues to increase, this artificial truncation is effectively mitigated, and the discrete eigenvalues converge toward the numerically converged reference spectrum. When we focus on the purple lines with a sufficiently expanded eigenbasis, their descent rate gradually slows down as the eigenvalue index $n$ increases. This slowing descent is more pronounced in small normalized apertures (e.g., $L = 1\lambda, 2\lambda$) than large normalized apertures (e.g., $L = 10\lambda$) because the observable transition window is limited by machine precision. 
This observation is qualitatively consistent with the dimensionality-folding behavior described in Remark~\ref{remark:decaying_rate}, which yields a sub-exponential-in-$n$ upper envelope under the flattened 2D eigenvalue ordering.
Furthermore, for small normalized apertures (e.g., $L = 1\lambda$ and $L = 2\lambda$), the eigenvalue distributions display noticeable ``staircase'' patterns. This reflects the eigenvalue degeneracy inherent to the geometric symmetries of the square aperture and the circular wavenumber domain. As the aperture size increases, the modal density intensifies, and the macroscopic eigenvalue distribution smooths out. 
	
\begin{figure}[htbp]
	\centerline{\includegraphics[scale=0.37]{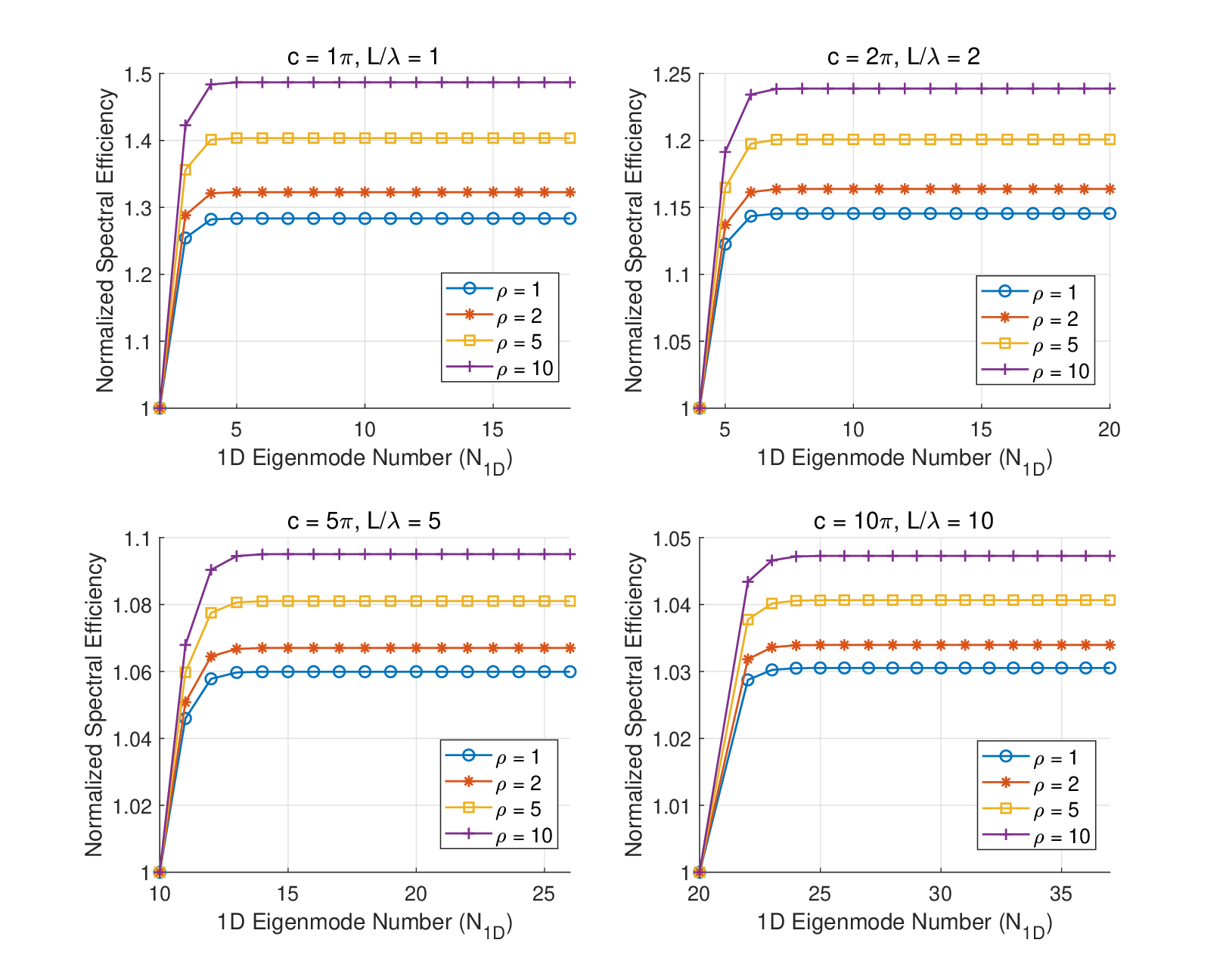}}
	\caption{Spectral Efficiency under various settings of normalized aperture size and SNR (i.e., $\rho$).}
	\label{Fig:Capacity_asymptotic}
\end{figure}

For a deterministic numerical evaluation of the spectral efficiency under different numbers of 1D PSWF basis functions, we consider an upper bound of the ergodic capacity in~\eqref{eq:actual_ergodic_capacity}.
Specifically, applying Jensen’s inequality to the concave log-determinant functional in~\eqref{eq:I_Q_Definition}, together with the Gaussian channel structure in~\eqref{eq:H_factored}, the total-power constraint $\operatorname{Tr}(\mathbf{Q})\leq P_T$, and \(\mu_0(\mathcal P)\le 1\), yields the following deterministic upper bound on the ergodic capacity~\cite{Poon2006Impact, Nam2014Capacity}:
\begin{equation}
	C_{\mathrm{upper}} = \sum_{n=0}^{\infty}\log_2(1 + \rho\mu_n(\mathcal{P})).
	\label{eq:C_upper}
\end{equation}
Given the approximation error bound established in Theorem~\ref{theorem:approx_error}, $C_{\mathrm{upper}}$ can be approached by the finite-dimensional quantity
\begin{equation}
	C_{\mathrm{upper}}^{(N)} = \sum_{n=0}^{N-1}\log_2(1 + \rho\mu_n(\mathbf{M}_N)).
	\label{eq:C_upper_N}
\end{equation}
Together with the elementary bound in~\eqref{eq:scalar_log_bound}, Theorem~\ref{theorem:approx_error} shows that $C_{\mathrm{upper}}^{(N)}$ converges to $C_{\mathrm{upper}}$ with the same certified super-exponential dependence on $N_{1D}$ once $N_{1D} > ec/4$.
 
Figure \ref{Fig:Capacity_asymptotic} illustrates the normalized spectral efficiency of the spatially continuous HMIMO channel as a function of the 1D eigenmode number $N_{1D}$ across various normalized aperture sizes ($L/\lambda \in \{1, 2, 5, 10\}$) and signal-to-noise ratios ($\rho \in \{1, 2, 5, 10\}$).
The spectral efficiencies are calculated using \eqref{eq:C_upper_N}, and the initial evaluated point of each curve is anchored at $N_{1D} = 2L/\lambda$, with the immediately subsequent point set to the 1D truncation threshold $N_{1D} = \lfloor ec/4 \rfloor + 1$. The spectral efficiencies are normalized with respect to the value of the first point of each curve.

As observed, increasing the number of basis functions yields a pronounced improvement in spectral efficiency, particularly for systems with small aperture sizes. Critically, once $N_{1D}$ reaches the derived threshold $\lfloor ec/4 \rfloor + 1$, the spectral efficiency becomes nearly saturated across all aperture configurations and SNR regimes.
This observed saturation is qualitatively consistent with the super-exponential truncation-order dependence of the actual ergodic-capacity gap established in Theorem~\ref{theorem:capacity_convergence}.
Furthermore, the numerical results highlight the substantial benefit of retaining spatial modes beyond the conventional 1D spatial DoF for small-aperture systems.
Specifically, for a highly compact aperture (e.g., $L/\lambda = 1$), the normalized spectral efficiency increases by nearly 40\% at a high SNR ($\rho = 10$) when increasing $N_{1D}$ from the 1D spatial DoF ($2L/\lambda$) to the analytical truncation threshold ($\lfloor ec/4 \rfloor + 1$). 
In contrast, for a larger aperture ($L/\lambda = 10$), the corresponding relative increase remains below 5\%. 
This contrast indicates that, for compact continuous apertures, truncating the spatial representation at the conventional $2L/\lambda$ benchmark may omit modes that make non-negligible contributions to the evaluated spectral efficiency.

\section{Conclusion}
\label{sec:conclusion}
In this paper, we characterized the eigenspectrum and eigenmodes of a non-separable continuous spatial-wavenumber integral operator by establishing a rigorous spectral approximation framework. 
We analytically demonstrated that, despite the loss of separability induced by the square--disk mismatch, the whole-spectrum approximation error remains controlled by the 1D PSWF eigenvalue tail and decays super-exponentially beyond the established 1D truncation threshold.
We further established an explicit non-asymptotic upper bound on the actual ergodic-capacity gap, which exhibits the same super-exponential dependence on the truncation order.
Numerical evaluations are consistent with these theoretical findings, further revealing that, for compact apertures, retaining modes beyond the conventional 1D spatial DoF $2L/\lambda$ can yield substantial improvements in the evaluated spectral efficiency.
Looking ahead, several directions remain for future research. Beyond rectangular apertures, extending the current framework to more general spatial geometries would be of interest.
Another challenging direction is to extend our framework to the near-field scenario, where the fundamental assumption of spatial wide-sense stationarity breaks down.

\appendices
\section{}
\label{appdix: PSWF_calculation_derivation}
A fundamental discovery in \cite{Slepian1961Prolate} is that the integral operator $\mathcal{Q}_c[\cdot](x)$ from \eqref{eq:operator_Q} commutes with the following second-order differential operator:
$$
	\mathcal{D}_x^c = -\frac{d}{dx}(1 - x^2)\frac{d}{dx} + c^2x^2.
	\label{eq:differential operator}
$$
This commutation property ensures that the PSWFs are also eigenfunctions of the singular Sturm-Liouville problem \cite{wang2017review}:
\begin{equation}
	\mathcal{D}_x^c\phi_n(x) = \chi_n\phi_n(x),	
	\label{eq:SLP}
\end{equation}
which arises from solving the Helmholtz equation in prolate spheroidal coordinates via separation of variables. Recall that the Legendre polynomials $\{P_m(x)\}$ satisfy the differential equation:
\begin{equation}
	\mathcal{D}_x^0 P_m(x) = -\frac{d}{dx}(1-x^2)\frac{dP_m(x)}{dx} = m(m+1)P_m(x).
	\label{eq:legendre_polynomials}
\end{equation}
The set $\{P_m(x)\}$ also obeys the three-term recurrence relation~\cite{szegö1975orthogonal}:
\begin{gather}
	(m+1)P_{m+1}(x) = (2m+1)xP_m(x) - mP_{m-1}(x), m \geq 1,\nonumber \\
	P_0(x) = 1, \; P_1(x) = x. \nonumber
	\label{eq:recurrence}
\end{gather}
Furthermore, we introduce the orthonormal Legendre polynomials, defined as:
\begin{align}
	\bar{P}_m(x) = \sqrt{\frac{2m+1}{2}}P_m(x)\; \text{so} \;
	\int_{-1}^{1}\bar{P}_m(x)\bar{P}_{m'}(x)dx = \delta_{mm'}.
	\label{eq: P_orthonormal}
\end{align}
By substituting the expansion \eqref{eq:Legendre_linear_combination} into \eqref{eq:SLP} and using the properties of Legendre polynomials from \eqref{eq:legendre_polynomials}–\eqref{eq: P_orthonormal}, we obtain the eigenvalue decomposition problem in \eqref{eq:eigenvalue_decomposition}.
	
\section{}
\label{appdix: PSWF_eigenvalue_computation}
PSWFs have the following properties: $\phi_n(x)$ is an even function if $n$ is even and an odd function if $n$ is odd. Furthermore, $\phi_n(x)$ has exactly $n$ real roots in $[-1, 1]$. 
For the specific case where $x = 0$, these properties imply that $\phi_n(x) \neq 0$ for even $n$, while $\phi_n(x) = 0$ and $\phi'_n(x) \neq 0$ for odd $n$. 

For the even case $(n = 2p)$, we let $x = 0$ in \eqref{eq:PSWF_definition} to get:
\begin{align}
	&i^{2p}|\lambda_{2p}|\phi_{2p}(0) = \int_{-1}^{1}\phi_{2p}(x')dx' = \nonumber \\ &\qquad\qquad\sqrt{2}\int_{-1}^{1}\phi_{2p}(x')\bar{P}_0(x')dx'
	= \sqrt{2}\beta_{2p, 0}.
	\label{eq:eigenvalue_even}
\end{align}
This derivation uses the properties $\lambda_n = i^n |\lambda_n|$ and $\bar{P}_0(x) = \sqrt{1/2}$, along with the relationship between the PSWF and the orthonormal Legendre polynomials:
$$
	\beta_{n, m} = \int_{-1}^{1}\phi_n(x)\bar{P}_m(x)dx.
	\label{eq:beta_n_k}
$$
Equation~\eqref{eq:eigenvalue_even} then leads to the expression in~\eqref{eq:lambda_computation_even}. 

For the odd case ($n = 2p+1$), we take the derivative of~\eqref{eq:PSWF_definition} with respect to $x$ and then set $x = 0$:
\begin{align}
	i^{2p+1}|\lambda_{2p+1}|\phi'_{2p+1}(0) = ic\int_{-1}^{1}x'\phi_{2p+1}(x')dx' \nonumber \\
	= \frac{ic}{\sqrt{3/2}}\int_{-1}^{1}\bar{P}_1(x')\phi_{2p+1}(x')dx' = \frac{ic \beta_{2p+1, 1}}{\sqrt{3/2}}.\nonumber
\end{align}
Here, we have used the fact that $\bar{P}_1(x) = \sqrt{3/2} x$. This result leads to the expression in~\eqref{eq:lambda_computation_odd}.

\section{}
\label{appdix:2D_PSWF_Calculation}
We first consider the inner double integral in (\ref{eq:M_pq_jl}) and rewrite it in the following convolutional form:
\begin{align}
	&\int_{-1}^{1}\int_{-1}^{1}G_c(x-x', y-y')\phi_j(x')\phi_\ell(y')\,dx'\,dy' \nonumber \\
 = & \int_{-\infty}^{\infty}\int_{-\infty}^{\infty}G_c(x-x', y-y')\tilde{\phi}_j(x')\tilde{\phi}_\ell(y')\,dx'\,dy'  \nonumber \\
	= & \; \tilde{G}_c(x, y) \ast \left\{\tilde{\phi}_j(x)\tilde{\phi}_\ell(y)\right\}.
	\label{eq:convol}
\end{align}
Here, the $\ast$ operator denotes 2D convolution. $\tilde{G}_c(x, y)$ is given by
\begin{equation}
	\tilde{G}_c(x, y) = \frac{1}{4\pi^2}\iint_{k_x^2 + k_y^2 \leq c^2}e^{i(k_x x + k_y y)}\,dk_x\,dk_y,
	\label{eq:low_pass_filter_2}
\end{equation}
and $\tilde{\phi}_n(x)$ is defined as
$$
	\tilde{\phi}_n(x) = \phi_n(x)\mathbb{I}_{[-1, 1]}(x).
	\label{eq:phi_tilde}
$$
Here, $\mathbb{I}_{[-a, a]}(x)$ is the indicator function for the interval $[-a, a]$, defined as
$$
	\mathbb{I}_{[-a, a]}(x) = 
	\left\{ \begin{array}{rcl}
		1 & \mbox{for}
		&  -a \leq x \leq a,   \\ 
		0 & & \text{otherwise}.
	\end{array}\right.
	\label{eq:indicator_function}
$$
By substituting (\ref{eq:convol}) into (\ref{eq:M_pq_jl}) and applying the indicator function to $\phi_p(x)$ and $\phi_q(y)$, we obtain the following result:
\begin{align}
	&M_{(pq)(j\ell)}  \nonumber \\ 
	=& \int_{-\infty}^{\infty}\int_{-\infty}^{\infty}\tilde{\phi}_p(x)\tilde{\phi}_q(y)\cdot \big\{\tilde{G}_c(x, y)\ast \big\{\tilde{\phi}_j(x)\tilde{\phi}_\ell(y)\big\}\big\}dxdy  \nonumber \\
	=& \frac{1}{4\pi^2}\int_{-\infty}^{\infty}\int_{-\infty}^{\infty}\mathcal{F}[\tilde{\phi}_p(x)\tilde{\phi}_q(y)](k_x, k_y) \nonumber \\
	&\times(\mathcal{F}[\tilde{\phi}_j(x)\tilde{\phi}_\ell(y)](k_x, k_y))^*
	(\mathcal{F}[\tilde{G}_c(x,y)](k_x, k_y))^*dk_xdk_y.
	\label{eq:M_pqjl_2}
\end{align}
Note that the second equality in (\ref{eq:M_pqjl_2}) follows from Parseval's identity. 
	
According to (\ref{eq:low_pass_filter_2}), $\mathcal{F}[\tilde{G}_c(x,y)](k_x, k_y)$ is an indicator function for a disk of radius $c$ in the wavenumber domain:
\begin{equation}
	\mathcal{F}[\tilde{G}_c(x,y)](k_x, k_y) = 
		\begin{cases} 
			1, & \text{if } k_x^2 + k_y^2 \leq c^2, \\ 
			0, & \text{otherwise}. 
		\end{cases}
		\label{eq:Fourier_transform_LPF}
\end{equation}

Furthermore, since the non-zero support of (\ref{eq:Fourier_transform_LPF}) inherently restricts the integration region to $k_x, k_y \in [-c, c]$, each 1D Fourier transform term in (\ref{eq:M_pqjl_2}) is strictly equivalent to its band-limited version within this domain (e.g., $\mathcal{F}[\tilde{\phi}_p](k_x)\mathbb{I}_{[-c, c]}(k_x)$).
As established in~\cite{osipov2013prolate}, this band-limited transform has the expression:
\begin{equation}
	\begin{aligned}
		&\mathcal{F}[\tilde{\phi}_p](k_x)\mathbb{I}_{[-c,c]}(k_x) = \\ &\qquad\quad
		\begin{cases} \lambda_p(-1)^p\phi_p(k_x/c), & \text{if } -c \leq k_x \leq c, \\ 
		0, & \text{otherwise}. \end{cases}
	\end{aligned}
	\label{eq:Fourier_PSWF_limit}
\end{equation}
Substituting (\ref{eq:Fourier_transform_LPF}) and (\ref{eq:Fourier_PSWF_limit}) into (\ref{eq:M_pqjl_2}) and using the property $\lambda_p = i^{p}|\lambda_p|$ \cite{osipov2013prolate}, we obtain:
\begin{align}
	&M_{(pq)(j\ell)} = \frac{\left[(-1)^{p+q+j+\ell}\right]\left(i^{p+q-j-\ell}\right)|\lambda_p||\lambda_q||\lambda_j||\lambda_\ell|}{4\pi^2} \nonumber \\
	&\times \iint_{\sqrt{k_x^2 + k_y^2} \leq c}\phi_p\Big(\frac{k_x}{c}\Big)\phi_q\Big(\frac{k_y}{c}\Big)\phi_j\Big(\frac{k_x}{c}\Big)\phi_\ell\Big(\frac{k_y}{c}\Big)dk_xdk_y .
	\label{eq:M_pqjl_3}
\end{align}
Thus, the original quadruple integral in (\ref{eq:M_pq_jl}) has been reduced to the double integral in (\ref{eq:M_pqjl_3}).

We can further simplify the integral in (\ref{eq:M_pqjl_3}) by exploiting the parity of its integrand.
Since a 1D PSWF, $\phi_n(x)$, is an even function for even $n$ and an odd function for odd $n$, the term $\phi_p(\frac{k_x}{c})\phi_j(\frac{k_x}{c})$ is an odd function of $k_x$ if $p+j$ is odd. Similarly, the term $\phi_q(\frac{k_y}{c})\phi_\ell(\frac{k_y}{c})$ is an odd function of $k_y$ if $q + \ell$ is odd. Because the integration domain (i.e., $k_x^2 + k_y^2 \leq c^2$) is symmetric with respect to both the $k_x$ and $k_y$ axes independently, the integral strictly vanishes if the integrand is odd in either $k_x$ or $k_y$. Therefore, we have:
\begin{equation}
	M_{(pq)(j\ell)} = 0 \; \text{for} \; \operatorname{mod}(p+j, 2) = 1 \; \text{or} \; \operatorname{mod}(q+\ell, 2) = 1.
	\label{eq:odd_function_2}
\end{equation}
On the other hand, if both $p+j$ and $q+\ell$ are even, the integrand is an even function with respect to both $k_x$ and $k_y$. In this case, the total sum $p+q+j+\ell$ is guaranteed to be even, meaning the term $(-1)^{p+q+j+\ell}$ simplifies to 1. $M_{(pq)(j\ell)}$ then becomes:
\begin{align}
	&M_{(pq)(j\ell)} =  \pm\frac{|\lambda_p||\lambda_q||\lambda_j||\lambda_\ell|}{4\pi^2} \nonumber \\
	& \times \iint_{\sqrt{k_x^2 + k_y^2} \leq c} \phi_p\Big(\frac{k_x}{c}\Big)\phi_q\Big(\frac{k_y}{c}\Big)\phi_j\Big(\frac{k_x}{c}\Big)\phi_\ell\Big(\frac{k_y}{c}\Big)\, dk_x \, dk_y \nonumber \\
	& \text{for} \quad \operatorname{mod}(p+j, 2) = 0 \quad \text{and} \quad \operatorname{mod}(q+\ell, 2)=0.
	\label{eq:even_function}
\end{align}
The sign in this expression is determined by the term $i^{p+q-j-\ell}$: it is positive if $\operatorname{mod}(p+q-j-\ell, 4) = 0$ and negative if $\operatorname{mod}(p+q-j-\ell, 4) = 2$. 
Finally, by using the relationship between $\{\mu_n\}$ and $\{\lambda_n\}$ from (\ref{eq:mu_n_definition}) and converting the Cartesian variables into polar coordinates ($dk_x \,dk_y = r \,dr\, d\theta$), we obtain the exact two-dimensional polar representation in (\ref{eq:even_function_2}).

\section{}
\label{appendix:truncation_error}
\subsubsection{Spectral Ordering and Trace Defect}

Recall that the eigenvalues $\{\mu_n(\mathcal{P})\}_{n=0}^{\infty}$
of the continuous concentration operator $\mathcal{P}$ in
\eqref{eq:integral_equation} are non-negative, since they are the
energy-concentration ratios defined in \eqref{eq:problem_2D_optimal_PSWF}.
Moreover, $\mathcal{P}$ is a compact self-adjoint integral operator,
and the infinite-dimensional matrix $\mathbf{M}$ in
\eqref{eq:linear_algebra} is its representation with respect to the
complete tensor-product PSWF basis. Therefore,
$$
	\mu_n(\mathbf{M})
	=
	\mu_n(\mathcal{P}),
	\qquad n=0,1,\ldots,
$$
as established in \eqref{eq:eigenvalue_equivalence}.
The truncated matrix $\mathbf{M}_N$ in \eqref{eq:EVD_N},
where $N=N_{1D}^2$, is the finite-dimensional principal truncation of
$\mathbf{M}$ corresponding to the retained tensor-product PSWFs indexed
by $0\leq p,q<N_{1D}$. 
By the max--min principle for compact self-adjoint operators~\cite[Theorem~4.22]{Borthwick2020Spectral}, projecting the operator to this finite-dimensional subspace yields
\begin{equation}
	0
	\leq
	\mu_n(\mathbf{M}_N)
	\leq
	\mu_n(\mathbf{M})
	=
	\mu_n(\mathcal{P}),
	\qquad 0\leq n<N.
	\label{eq:Ritz_ordering}
\end{equation}
Consequently, the retained spectral differences are all
non-negative, and hence
\begin{align}
	&\sum_{n=0}^{N-1}
	\left|
	\mu_n(\mathcal{P})-\mu_n(\mathbf{M}_N)
	\right|
	+
	\sum_{n=N}^{\infty}\mu_n(\mathcal{P})
	\nonumber\\
	&=
	\sum_{n=0}^{\infty}\mu_n(\mathcal{P})
	-
	\sum_{n=0}^{N-1}\mu_n(\mathbf{M}_N) =
	\operatorname{Tr}(\mathcal{P})
	-
	\operatorname{Tr}(\mathbf{M}_N).
	\label{eq:trace_defect_identity}
\end{align}

\subsubsection{Trace Defect and Dimensionality Reduction}
Since the tensor-product PSWFs form a complete orthonormal
basis, the trace of $\mathcal{P}$ is equal to the sum of the
diagonal elements of its infinite-dimensional matrix
representation $\mathbf{M}$. Hence,
\begin{align}
	\operatorname{Tr}(\mathcal{P})
	-
	\operatorname{Tr}(\mathbf{M}_N) =
	\sum_{\max(p,q)\geq N_{1D}}
	M_{(pq)(pq)}.
	\label{eq:trace_defect_diagonal}
\end{align}

Setting $j=p$ and $\ell=q$ in
\eqref{eq:even_function_2}, the sign is always positive and
the diagonal matrix element becomes
\begin{align}
	&M_{(pq)(pq)}\nonumber \\
	&=
	\frac{\mu_p\mu_q}{c^2}
	\iint_{k_x^2+k_y^2\leq c^2}
	\phi_p^2\left(\frac{k_x}{c}\right)
	\phi_q^2\left(\frac{k_y}{c}\right)
	dk_xdk_y.
	\label{eq:M_diagonal}
\end{align}
Because the disk $k_x^2+k_y^2\leq c^2$ is contained in the
square $[-c,c]^2$ and the integrand in
\eqref{eq:M_diagonal} is non-negative, we have
\begin{align}
	&\iint_{k_x^2+k_y^2\leq c^2}
	\phi_p^2\left(\frac{k_x}{c}\right)
	\phi_q^2\left(\frac{k_y}{c}\right)
	dk_x \, dk_y
	\nonumber\\
	& \leq\iint_{[-c,c]^2}
	\phi_p^2\left(\frac{k_x}{c}\right)
	\phi_q^2\left(\frac{k_y}{c}\right)
	dk_x \, dk_y\nonumber \\
	&= \left[\int_{-c}^{c} \phi_p^2\left(\frac{k_x}{c}\right)dk_x\right]\left[\int_{-c}^{c} \phi_q^2\left(\frac{k_y}{c}\right)\, dk_y\right]
	=c^2, \nonumber
\end{align}
where the last equality follows from the orthonormality of the
1D PSWFs on $[-1,1]$. Consequently, we have
\begin{equation}
	0\leq M_{(pq)(pq)}
	\leq \mu_p\mu_q.
	\label{eq:M_diagonal_bound}
\end{equation}

Substituting \eqref{eq:M_diagonal_bound} into
\eqref{eq:trace_defect_diagonal} yields
$$
	\operatorname{Tr}(\mathcal{P})
	-
	\operatorname{Tr}(\mathbf{M}_N)
	\leq
	\sum_{\max(p,q)\geq N_{1D}}\mu_p\mu_q.
	\label{eq:trace_defect_tail_1}
$$
The index set satisfying $\max(p,q)\geq N_{1D}$ is the
complement of the retained square
$\{(p,q):0\leq p,q<N_{1D}\}$ in the full non-negative
index plane. Therefore,
\begin{align}
	\sum_{\max(p,q)\geq N_{1D}}\mu_p\mu_q
	&=
	\sum_{p=0}^{\infty}\sum_{q=0}^{\infty}\mu_p\mu_q
	-
	\sum_{p=0}^{N_{1D}-1}
	\sum_{q=0}^{N_{1D}-1}\mu_p\mu_q
	\nonumber\\
	&=
	\left(\sum_{p=0}^{\infty}\mu_p\right)^2
	-
	\left(\sum_{p=0}^{N_{1D}-1}\mu_p\right)^2.
	\label{eq:trace_defect_tail_2}
\end{align}
Noting that $\sum_{p=0}^{N_{1D}-1}\mu_p = \sum_{p=0}^{\infty}\mu_p - \sum_{p=N_{1D}}^\infty\mu_p$, and applying the difference-of-squares identity, the right-hand side of
\eqref{eq:trace_defect_tail_2} can be rewritten as
$$
	2
	\left(\sum_{p=0}^{\infty}\mu_p\right)
	\left(\sum_{p=N_{1D}}^{\infty}\mu_p\right)
	-
	\left(\sum_{p=N_{1D}}^{\infty}\mu_p\right)^2\leq
	\frac{4c}{\pi}
	\sum_{p=N_{1D}}^{\infty}\mu_p,
$$
where the inequality is because
$\sum_{p=0}^{\infty}\mu_p=2c/\pi$
from~\cite{Landau1962PSWF_PartIII} and the non-negativity of the tail sum \(\sum_{p=N_{1D}}^{\infty}\mu_p\).
Therefore,
\begin{equation}
	\operatorname{Tr}(\mathcal{P})
	-
	\operatorname{Tr}(\mathbf{M}_N)
	\leq
	\frac{4c}{\pi}
	\sum_{p=N_{1D}}^{\infty}\mu_p.
	\label{eq:trace_defect_1D_tail}
\end{equation}

\subsubsection{Strict Non-Asymptotic Envelope of the 1D Tail}

According to Theorems 3.19 and 3.20 of
\cite{osipov2013prolate}, the eigenvalue $\lambda_n(c)$ defined
in \eqref{eq:PSWF_definition} satisfies
$$
	|\lambda_n(c)|
	\leq
	\nu(n,c),
$$
where
\begin{equation}
	\nu(n,c)
	=
	\frac{\sqrt{\pi}\,c^n(n!)^2}
	{(2n)!\Gamma(n+3/2)}.
	\label{eq:nu_n_c}
\end{equation}
Using $\Gamma\left(n+\frac{3}{2}\right)=\frac{(2n+2)!\sqrt{\pi}}{4^{n+1}(n+1)!}$
and $\binom{2n}{n}=\frac{(2n)!}{(n!)^2}$,
\eqref{eq:nu_n_c} can be rewritten exactly as
\begin{equation}
	\nu(n,c)
	=
	\frac{2(4c)^n}
	{(2n+1)n!\binom{2n}{n}^{2}}.
	\label{eq:nu_n_c_2}
\end{equation}

For every integer $n\geq1$, 
the Wallis-product bound
in~\cite[Eq.~(3)]{Hirschhorn2015WallissPA} gives
\begin{equation}
	\binom{2n}{n}
	>
	\frac{4^n}
	{\sqrt{\pi(n+1/2)}}.
	\label{eq:binomial_Wallis}
\end{equation}
Substituting \eqref{eq:binomial_Wallis} into
\eqref{eq:nu_n_c_2} yields
$$
	\nu(n,c)
	<
	\frac{\pi}{n!}
	\left(\frac{c}{4}\right)^n.
	\label{nu_n_c_upperbound}
$$
Furthermore, the strict Stirling bound of Robbins
\cite{robbins1955remark} gives $n! > \sqrt{2\pi n}\left(\frac{n}{e}\right)^n$.
We therefore obtain:
$$
	|\lambda_n(c)|
	<
	\sqrt{\frac{\pi}{2n}}
	\left(\frac{ce}{4n}\right)^n.
	\label{lambda_n_upperbound}
$$

Using the relationship between $\mu_n$ and $\lambda_n$ in
\eqref{eq:mu_n_definition}, we consequently obtain the strict
non-asymptotic eigenvalue envelope
\begin{equation}
	\mu_n(c)
	<
	\frac{c}{4n}
	\left(\frac{ce}{4n}\right)^{2n}
	=
	\frac{c}{4n}
	e^{-2n\ln\left(\frac{4n}{ce}\right)}.
	\label{eq:inequality_n_c}
\end{equation}
It is clear from~\eqref{eq:inequality_n_c} that $\mu_n(c)$ enters the super-exponential decay regime when the following threshold condition is satisfied:
$$
	n > \frac{ec}{4}.
	\label{eq:n_threshold}
$$

Substituting \eqref{eq:inequality_n_c} into \eqref{eq:trace_defect_1D_tail}, we now evaluate the corresponding infinite series. For $p\geq N_{1D}$, we define
$$
	a_p
	\triangleq
	\frac{c}{4p}
	\left(\frac{ec}{4p}\right)^{2p}.
$$
Then
$$
	\sum_{p=N_{1D}}^{\infty}\mu_p
	<
	\sum_{p=N_{1D}}^{\infty}a_p.
	\label{eq:mu_tail_ap}
$$
The ratio of two successive terms is
$$
	\frac{a_{p+1}}{a_p}
	=
	\frac{p}{p+1}
	\left(\frac{p}{p+1}\right)^{2p}
	\left(
	\frac{ec}{4(p+1)}
	\right)^2.
$$
Since $p/(p+1)<1$, the ratio is bounded by:
$$
	\frac{a_{p+1}}{a_p} <
	\left(
	\frac{ec}{4(p+1)}
	\right)^2 \leq
	\left(
	\frac{ec}{4(N_{1D}+1)}
	\right)^2
	\triangleq
	q(N_{1D},c).
	\label{eq:q_definition}
$$
Therefore, $a_{N_{1D} + m}$ can be upper-bounded by the geometric progression:
$$
	a_{N_{1D} + m}\leq a_{N_{1D}}(q(N_{1D}, c))^m.
$$
When $N_{1D}>ec/4$, we have $0<q(N_{1D},c)<1$.
Consequently,
\begin{align}
	\sum_{p=N_{1D}}^{\infty}\mu_p
	< a_{N_{1D}}\sum_{m=0}^\infty q(N_{1D}, c)^m = 
	\frac{a_{N_{1D}}}
	{1-q(N_{1D},c)}.
	\label{eq:infinite_series_inequality}
\end{align}

\subsubsection{Conclusion}
Combining \eqref{eq:trace_defect_identity},
\eqref{eq:trace_defect_1D_tail}, and
\eqref{eq:infinite_series_inequality}, we obtain
\begin{align}
	&\sum_{n=0}^{N-1}
	\left|
	\mu_n(\mathcal{P})
	-
	\mu_n(\mathbf{M}_N)
	\right|
	+
	\sum_{n=N}^{\infty}\mu_n(\mathcal{P})
	\nonumber\\
	&\quad<
	\frac{c^2}{\pi N_{1D}}
	\left[
	1-
	\left(
	\frac{ec}{4(N_{1D}+1)}
	\right)^2
	\right]^{-1} \cdot e^{-2N_{1D}
		\ln\left(\frac{4N_{1D}}{ce}\right)}.
	\label{eq:trace_defect_final}
\end{align}
By the definition of $C_1(N_{1D},c)$ in
\eqref{eq:C1_definition}, \eqref{eq:trace_defect_final}
is precisely the bound in
Theorem~\ref{theorem:approx_error}, which completes the proof.

\section{}
\label{appendix:channel_representations}

\subsubsection{Proof of Lemma~\ref{lemma:continuous_capacity}}
By Mercer's theorem, the receive and transmit spatial autocorrelation
functions in \eqref{eq:Rr_Gc} and \eqref{eq:Rt_Gc} admit the spectral
decompositions
$$
	R_r(\mathbf{r}_r,\mathbf{r}'_r)
	=
	\sum_{n=0}^{\infty}
	\mu_n(\mathcal{P})
	\phi_{r,n}(\mathbf{r}_r)
	\phi_{r,n}^{*}(\mathbf{r}'_r),
	\label{eq:Rr_mercer}
$$
$$
	R_t(\mathbf{r}_t,\mathbf{r}'_t)
	=
	\sum_{m=0}^{\infty}
	\mu_m(\mathcal{P})
	\phi_{t,m}(\mathbf{r}_t)
	\phi_{t,m}^{*}(\mathbf{r}'_t),
	\label{eq:Rt_mercer}
$$
where $\{\phi_{r,n}\}_{n=0}^{\infty}$ and
$\{\phi_{t,m}\}_{m=0}^{\infty}$ denote the corresponding
orthonormal spatial eigenfunctions at the receiver and transmitter,
respectively. Since the channel is assumed to be zero-mean proper complex Gaussian with the separable correlation structure in
\eqref{eq:separable_correlation}, the Karhunen--Lo\`eve expansion is
$$
	h(\mathbf{r}_r,\mathbf{r}_t)
	=
	\sum_{n=0}^{\infty}
	\sum_{m=0}^{\infty}
	w_{n,m}
	\sqrt{\mu_n(\mathcal{P})\mu_m(\mathcal{P})}
	\phi_{r,n}(\mathbf{r}_r)
	\phi_{t,m}^{*}(\mathbf{r}_t),
	\label{eq:h_infinite_series}
$$
where $\{w_{n,m}\}$ are i.i.d. $\mathcal{CN}(0,1)$ random variables.

We then express the transmit field, the received field, and the noise
field using their respective spatial basis functions as
$X(\mathbf{r}_t)=\sum_{m=0}^{\infty}x_m\phi_{t,m}(\mathbf{r}_t)$,
$Y(\mathbf{r}_r)=\sum_{n=0}^{\infty}y_n\phi_{r,n}(\mathbf{r}_r)$,
and
$Z(\mathbf{r}_r)=\sum_{n=0}^{\infty}z_n\phi_{r,n}(\mathbf{r}_r)$.
Substituting these expansions into \eqref{eq:signal_model} and
exploiting the orthonormality of both the transmit and receive basis
functions yields
\begin{equation}
	y_n
	=
	\sum_{m=0}^{\infty}
	\left(
	\sqrt{\mu_{r,n}}\,w_{n,m}\sqrt{\mu_{t,m}}
	\right)x_m+z_n,
	\qquad n\geq0.
	\label{eq:y_n_infty}
\end{equation}
Using \eqref{eq:symmetric_covariance_eigenvalues},
\eqref{eq:y_n_infty} gives the equivalent channel representation in
\eqref{eq:equivalent_channel} and the factorization in
\eqref{eq:H_factored}.

The equivalent channel $\mathbf{H}$ is well defined as a Hilbert--Schmidt random operator. In fact, we have $\sum_{n=0}^{\infty}\mu_n(\mathcal{P})
=
\operatorname{Tr}(\mathcal{P})
=
\frac{c^2}{\pi}$.
Using~\eqref{eq:H_factored}, we can obtain:
$$
	\mathbb{E}\left[\|\mathbf{H}\|_{\mathrm{HS}}^2\right]
	=
	\left(\sum_{n=0}^{\infty}\mu_n(\mathcal{P})\right)^2
	=
	\left(\frac{c^2}{\pi}\right)^2
	<\infty.
	\label{eq:H_HS}
$$
Hence, $\mathbf{H}$ is Hilbert--Schmidt almost surely. For every feasible $\mathbf{Q}\succeq0$, we have
$\|\mathbf{Q}\|\leq\operatorname{Tr}(\mathbf{Q})\leq P_T$. Moreover,
$$
	\operatorname{Tr}
	\left(
	\mathbf{H}\mathbf{Q}\mathbf{H}^{\dagger}
	\right)
	=
	\left\|
	\mathbf{H}\mathbf{Q}^{1/2}
	\right\|_{\mathrm{HS}}^2
	\leq
	\|\mathbf{H}\|_{\mathrm{HS}}^2
	\|\mathbf{Q}\|
	\leq
	P_T\|\mathbf{H}\|_{\mathrm{HS}}^2.
	\label{eq:HQH_trace_bound}
$$
Therefore, we have
$$
	\mathbb{E}
	\left[
	\operatorname{Tr}
	\left(
	\mathbf{H}\mathbf{Q}\mathbf{H}^{\dagger}
	\right)
	\right]
	\leq
	P_T
	\mathbb{E}
	\left[
	\|\mathbf{H}\|_{\mathrm{HS}}^2
	\right]
	=
	P_T
	\left(
	\frac{c^2}{\pi}
	\right)^2
	<\infty.
	\label{eq:HQH_expectation_bound}
$$
Thus $\mathbf{H}\mathbf{Q}\mathbf{H}^{\dagger}$ is positive trace
class almost surely. Together with
$\log_2\det(\mathbf{I}+\mathbf{A})\leq
\operatorname{Tr}(\mathbf{A})/\ln2$ for $\mathbf{A}\succeq0$, this
also verifies the Fredholm well-posedness and finiteness of the
expectation in~\eqref{eq:I_Q_Definition}.

We next show that the capacity optimization can be restricted to
diagonal transmit covariance matrices. Following the
covariance-symmetrization argument for right-symmetric MIMO
channels~\cite[Theorem~4.2]{RheeTaricco2006}, let
$$
	\mathbf{D}_{\boldsymbol{\theta}}
	\triangleq
	\operatorname{diag}
	\left(e^{i\theta_0},e^{i\theta_1},\ldots\right)
$$
be an arbitrary diagonal unitary matrix. Since
$\mathbf{D}_{\boldsymbol{\theta}}$ and $\boldsymbol{\Lambda}$ in \eqref{eq:Lambda_definition} are
both diagonal, they commute. Furthermore, because the entries of
$\mathbf{H}_w$ are independent circularly symmetric complex Gaussian
random variables, we have $\mathbf{H}_w\mathbf{D}_{\boldsymbol{\theta}}
	\overset{d}{=}
	\mathbf{H}_w$,
where $\overset{d}{=}$ denotes equality in distribution.

We have $\mathcal{I}(\mathbf{Q})$ as defined in~\eqref{eq:I_Q_Definition}. Therefore, for any feasible $\mathbf{Q}$, using the property $\mathbf{H}_w\mathbf{D}_{\boldsymbol{\theta}}
\overset{d}{=}
\mathbf{H}_w$, we have
\begin{align}
	&\mathcal{I}
	\left(
	\mathbf{D}_{\boldsymbol{\theta}}
	\mathbf{Q}
	\mathbf{D}_{\boldsymbol{\theta}}^{\dagger}
	\right) =
	\nonumber\\
	&
	\mathbb{E}
	\left[
	\log_2\det
	\left(
	\mathbf{I}
	+
	\frac{1}{\sigma_z^2}
	\boldsymbol{\Lambda}^{1/2}
	\mathbf{H}_w
	\mathbf{D}_{\boldsymbol{\theta}}
	\boldsymbol{\Lambda}^{1/2}
	\mathbf{Q}
	\boldsymbol{\Lambda}^{1/2}
	\mathbf{D}_{\boldsymbol{\theta}}^{\dagger}
	\mathbf{H}_w^{\dagger}
	\boldsymbol{\Lambda}^{1/2}
	\right)
	\right]
	\nonumber\\
	&=
	\mathcal{I}(\mathbf{Q}).
	\label{eq:I_phase_invariance}
\end{align}
Let the phases $\{\theta_m\}_{m=0}^{\infty}$ be independent and
uniformly distributed over $[0,2\pi)$. Averaging the phase-rotated
covariance matrix gives
$$
	\mathbb{E}_{\boldsymbol{\theta}}
	\left[
	\mathbf{D}_{\boldsymbol{\theta}}\mathbf{Q}
	\mathbf{D}_{\boldsymbol{\theta}}^{\dagger}
	\right]
	=
	\operatorname{diag}(Q_{0,0},Q_{1,1},\ldots)
	=
	\widetilde{\mathbf{Q}}.
	\label{eq:Q_phase_average}
$$
Indeed, the $(m,\ell)$-th entry of the left-hand side is $Q_{m,\ell}
\mathbb{E}_{\boldsymbol{\theta}}
\left[
e^{i(\theta_m-\theta_\ell)}
\right]$, which equals $Q_{m,m}$ for $m=\ell$ and vanishes for
$m\neq\ell$.

Since $\mathbf{Q}\succeq0$, the phase averaging preserves feasibility:
$$
	\widetilde{\mathbf{Q}}\succeq0,
	\qquad
	\operatorname{Tr}(\widetilde{\mathbf{Q}})
	=
	\operatorname{Tr}(\mathbf{Q})
	\leq P_T.
	\label{eq:Q_tilde_feasible}
$$
The functional $\mathcal{I}(\mathbf{Q})$ is concave in
$\mathbf{Q}\succeq0$. Hence, applying Jensen's inequality with
respect to the auxiliary random phases $\boldsymbol{\theta}$ yields
$$
	\mathcal{I}(\tilde{\mathbf{Q}}) = \mathcal{I}
	\left(
	\mathbb{E}_{\boldsymbol{\theta}}
	\left[
	\mathbf{D}_{\boldsymbol{\theta}}
	\mathbf{Q}
	\mathbf{D}_{\boldsymbol{\theta}}^{\dagger}
	\right]
	\right) \geq
	\mathbb{E}_{\boldsymbol{\theta}}
	\left[
	\mathcal{I}
	\left(
	\mathbf{D}_{\boldsymbol{\theta}}
	\mathbf{Q}
	\mathbf{D}_{\boldsymbol{\theta}}^{\dagger}
	\right)
	\right]
	=
	\mathcal{I}(\mathbf{Q}),
	\label{eq:Q_diagonalization}
$$
where the last equality follows from \eqref{eq:I_phase_invariance}.
Therefore, for every feasible transmit covariance matrix $\mathbf{Q}$,
there exists a feasible diagonal covariance matrix
$\widetilde{\mathbf{Q}}$ whose ergodic rate is no smaller, which
yields \eqref{eq:ergodic_capacity_diagonal_Q}. It is important to emphasize that the argument relies on
the statistical-CSIT assumption, under which the transmit covariance
is independent of the instantaneous realization of $\mathbf{H}_w$. 

The above subsection provides the detailed derivation of the discrete representation in~\eqref{eq:equivalent_channel} and proves Lemma~\ref{lemma:continuous_capacity}.

\subsubsection{Proof of Lemma~\ref{lemma:truncated_capacity}}

Recall that the retained spatial subspace is spanned by the
tensor-product PSWFs
$\{\phi_p(x)\phi_q(y)\}_{0\leq p,q<N_{\mathrm{1D}}}$,
with $N=N_{\mathrm{1D}}^2$. For notational convenience, let $\psi_{pq}(\mathbf r)
\triangleq
\phi_p(x)\phi_q(y)$ with $\mathbf{r} = (x, y)$.

Projecting the continuous channel kernel
$h(\mathbf r_r,\mathbf r_t)$ onto this subspace at both the
transmit and receive sides gives the $N$-dimensional random
channel matrix $\mathbf H_N^{\mathrm{PSWF}}$, whose
$((pq),(j\ell))$-th entry is
$$
	\label{eq:projected_channel_element}
	[\mathbf H_N^\mathrm{PSWF}]_{(pq),(j\ell)}
	=
	\int_{\mathcal D_r}
	\int_{\mathcal D_t}
	\psi_{pq}^{*}(\mathbf r_r)
	h(\mathbf r_r,\mathbf r_t)
	\psi_{j\ell}(\mathbf r_t)
	\,d\mathbf r_t\,d\mathbf r_r .
$$
Since $h(\mathbf r_r,\mathbf r_t)$ is a zero-mean proper complex Gaussian random field, the entries of $\mathbf H_N^{\mathrm{PSWF}}$, being linear functionals of $h$,
are jointly zero-mean proper complex Gaussian random variables. We next determine their covariance. For two arbitrary pairs of
transmit and receive indices, we have
\begin{align}
	&\mathbb E\!
	\left[
	[\mathbf H_N^{\mathrm{PSWF}}]_{(pq),(j\ell)}
	[\mathbf H_N^{\mathrm{PSWF}}]_{(p'q'),(j'\ell')}^{*}
	\right]
	\nonumber\\
	&=
	\int_{\mathcal D_r}
	\int_{\mathcal D_r}
	\int_{\mathcal D_t}
	\int_{\mathcal D_t}
	\psi_{pq}^{*}(\mathbf r_r)
	\psi_{p'q'}(\mathbf r_r')
	\psi_{j\ell}(\mathbf r_t)
	\psi_{j'\ell'}^{*}(\mathbf r_t')
	\nonumber\\
	&\qquad\times
	\mathbb E
	\left[
	h(\mathbf r_r,\mathbf r_t)
	h^{*}(\mathbf r_r',\mathbf r_t')
	\right]
	\,d\mathbf r_t\,d\mathbf r_t'\,
	d\mathbf r_r\,d\mathbf r_r'.
	\label{eq:projected_channel_covariance_1}
\end{align}
Given the separable channel correlation model in~\eqref{eq:separable_correlation}, \eqref{eq:projected_channel_covariance_1} separates
into the receive- and transmit-side terms as
\begin{align}
	&\mathbb E\!
	\left[
	[\mathbf H_N^{\mathrm{PSWF}}]_{(pq),(j\ell)}
	[\mathbf H_N^{\mathrm{PSWF}}]_{(p'q'),(j'\ell')}^{*}
	\right]
	\nonumber\\
	&= 
	[\mathbf R_{r,N}]_{(pq),(p'q')}
	[\mathbf R_{t,N}]_{(j\ell),(j'\ell')}^{*},
	\label{eq:projected_channel_covariance_2}
\end{align}
where
\begin{align}
	&[\mathbf R_{r,N}]_{(pq),(p'q')}\nonumber \\
	&\quad \triangleq
	\int_{\mathcal D_r}\int_{\mathcal D_r}
	\psi_{pq}^{*}(\mathbf r_r)
	R_r(\mathbf r_r,\mathbf r_r')
	\psi_{p'q'}(\mathbf r_r')
	\,d\mathbf r_r'\,d\mathbf r_r ,
	\label{eq:projected_receive_covariance}\\
	&[\mathbf R_{t,N}]_{(j\ell),(j'\ell')}
	\triangleq
	\int_{\mathcal D_t}\int_{\mathcal D_t}
	\psi_{j\ell}^{*}(\mathbf r_t)
	R_t(\mathbf r_t,\mathbf r_t')
	\psi_{j'\ell'}(\mathbf r_t')
	\,d\mathbf r_t'\,d\mathbf r_t .
	\label{eq:projected_transmit_covariance}
\end{align}

For the symmetric normalized square--disk channel considered in
this paper, $R_r=R_t=G_c$. Since the tensor-product PSWFs are real-valued, comparing
\eqref{eq:projected_receive_covariance} and
\eqref{eq:projected_transmit_covariance} with the definition of
$M_{(pq)(j\ell)}$ in \eqref{eq:M_pq_jl}, we obtain
$\mathbf R_{r,N} = \mathbf R_{t,N} = \mathbf M_N$.
Since $\mathbf M_N$ is real symmetric, it follows from
\eqref{eq:projected_channel_covariance_2} that
\begin{align}
	&\mathbb E\!
	\left[
	[\mathbf H_N^{\mathrm{PSWF}}]_{(pq),(j\ell)}
	[\mathbf H_N^{\mathrm{PSWF}}]_{(p'q'),(j'\ell')}^{*}
	\right]\nonumber \\
	& =
	[\mathbf M_N]_{(pq),(p'q')}
	[\mathbf M_N]_{(j\ell),(j'\ell')}. \nonumber
	\label{eq:projected_channel_covariance_MN}
\end{align}
Now let $\mathbf H_{w,N}\in\mathbb C^{N\times N}$ have i.i.d. $\mathcal{CN}(0,1)$ entries and define
\begin{equation}
	\widehat{\mathbf H}_N
	\triangleq
	\mathbf M_N^{1/2}
	\mathbf H_{w,N}
	\mathbf M_N^{1/2}. \nonumber
	\label{eq:projected_channel_auxiliary}
\end{equation}
Since \(\mathbf M_N\succeq0\), its square root is well defined. As a deterministic linear transformation of \(\mathbf H_{w,N}\), \(\widehat{\mathbf H}_N\) is a zero-mean proper complex Gaussian random matrix.
Moreover, since $\mathbf H_{w,N}$ has i.i.d. $\mathcal{CN}(0, 1)$ entries, we have
\begin{align}
	&\mathbb E
	\left[
	[\widehat{\mathbf H}_N]_{(pq),(j\ell)}
	[\widehat{\mathbf H}_N]_{(p'q'),(j'\ell')}^{*}
	\right]
	\nonumber\\
	&=
	[\mathbf M_N]_{(pq),(p'q')}
	[\mathbf M_N]_{(j\ell),(j'\ell')}. \nonumber
	\label{eq:auxiliary_channel_covariance}
\end{align}
Thus, $\widehat{\mathbf H}_N$ and
$\mathbf H_N^{\mathrm{PSWF}}$ are both zero-mean proper complex
Gaussian random matrices with identical covariance matrices.
Consequently, they have the same distribution:
$$
		\mathbf H_N^{\mathrm{PSWF}}
		\overset{d}{=}
		\mathbf M_N^{1/2}
		\mathbf H_{w,N}
		\mathbf M_N^{1/2}.
	\label{eq:HN_PSWF_distribution}
$$

Since the spatial covariance matrix $\mathbf{M}_N$ can be eigen-decomposed as $\mathbf{M}_N
=
\mathbf{U}_N
\boldsymbol{\Lambda}_N
\mathbf{U}_N^\dagger$, expressing the projected channel in the eigenbasis of $\mathbf{M}_N$ gives
$$
	\mathbf{U}_N^\dagger
	\mathbf{H}_N^{\mathrm{PSWF}}
	\mathbf{U}_N
	=
	\boldsymbol{\Lambda}_N^{1/2}
	\left(
	\mathbf{U}_N^\dagger
	\mathbf{H}_{w,N}
	\mathbf{U}_N
	\right)
	\boldsymbol{\Lambda}_N^{1/2}.
$$
Because $\mathbf{H}_{w,N}$ has i.i.d. $\mathcal{CN}(0, 1)$ entries, its distribution is invariant under deterministic unitary
transformations from the left and right. Hence,
$$
	\mathbf{U}_N^\dagger
	\mathbf{H}_{w,N}
	\mathbf{U}_N
	\overset{d}{=}
	\mathbf{H}_{w,N}.
$$
Accordingly, without loss of distribution, the projected channel $\mathbf{H}_N^{\mathrm{PSWF}}$ can
be represented in the eigenbasis of $\mathbf{M}_N$ as~\eqref{eq:HN_eigenbasis}. 

Applying the same phase-invariance and
concavity argument established in the proof of
Lemma~\ref{lemma:continuous_capacity} to this finite-dimensional
channel shows that the maximization in
\eqref{eq:finite_ergodic_capacity} can be restricted, without loss of
optimality, to diagonal transmit covariance matrices $\widetilde{\mathbf{Q}}_N$, yielding
\eqref{eq:finite_capacity_diagonal_Q}. The above subsection gives the derivation for the projected channel representation in~\eqref{eq:HN_eigenbasis} and proves
Lemma~\ref{lemma:truncated_capacity}.

\section{}
\label{appendix:capacity_convergence}

We prove Theorem~\ref{theorem:capacity_convergence} by comparing
the continuous and PSWF-truncated channels for a common admissible
power-allocation sequence. Throughout this appendix, let $\widetilde{\mathbf Q}
=
\operatorname{diag}(q_0,q_1,\ldots),\,
q_m\geq0$ satisfy
\begin{equation}
	\sum_{m=0}^{\infty}q_m
	=
	\operatorname{Tr}(\widetilde{\mathbf Q})
	\leq P_T .
	\label{eq:appendixF_power}
\end{equation}
For the $N$-dimensional truncated channel, only the first $N$
entries $q_0,\ldots,q_{N-1}$ are relevant.

Recall that $\boldsymbol{\Lambda}$ denotes the common transmit-
and receive-side covariance eigenvalue matrix of the continuous
symmetric channel as defined in
\eqref{eq:Lambda_definition}. In contrast,
$\boldsymbol{\Lambda}_N$ retains its
finite-dimensional definition in
\eqref{eq:Lambda_N_definition}. To compare the continuous and
truncated spectra on the common index set $n=0,1,\ldots$, we
extend only the eigenvalue sequence of $\mathbf M_N$ by the
convention
\begin{equation}
	\mu_n(\mathbf M_N)\triangleq0,
	\qquad n\geq N.
	\label{eq:MN_zero_extension}
\end{equation}
By the eigenvalue ordering established in \eqref{eq:Ritz_ordering} of
Appendix~\ref{appendix:truncation_error}, this convention gives
\begin{equation}
	0
	\leq
	\mu_n(\mathbf M_N)
	\leq
	\mu_n(\mathcal P),
	\qquad n=0,1,\ldots ,
	\label{eq:capacity_spectral_ordering}
\end{equation}
and
\begin{equation}
	\sum_{n=0}^{\infty}
	\left[
	\mu_n(\mathcal P)-\mu_n(\mathbf M_N)
	\right]
	=
	\operatorname{Tr}(\mathcal P)
	-
	\operatorname{Tr}(\mathbf M_N).
	\label{eq:trace_equivalence}
\end{equation}
The comparison below is carried out in the respective covariance
eigenbases of the continuous and truncated channels, with their
eigenvalues paired according to the common non-increasing ordering
used in Theorem~\ref{theorem:approx_error}. We use the same
Gaussian array $\{w_{n,m}\}_{n,m\geq0}$ throughout the proof;
the finite matrix $\mathbf H_{w,N}$ in~\eqref{eq:HN_eigenbasis} is identified
with its leading $N\times N$ block.

\subsubsection{An Auxiliary Log-Determinant Bound}

We first establish an inequality used in both the receive- and
transmit-side comparisons. Let $\mathbf{A}$ and $\mathbf{B}$ be
positive trace-class operators satisfying $\mathbf{A}\succeq\mathbf{B}\succeq\mathbf{0}$.
By the min--max principle, their ordered eigenvalues satisfy
$$
	\lambda_n(\mathbf{A})
	\geq
	\lambda_n(\mathbf{B}),
	\qquad n=0,1,\ldots
$$
Moreover, due to the concavity of $\log_2(1 + x)$, for $x\geq y\geq0$, we have
\begin{equation}
	0
	\leq
	\log_2(1+x)-\log_2(1+y)
	\leq
	\frac{x-y}{\ln2}.
	\label{eq:scalar_log_bound}
\end{equation}
Therefore, we can obtain:
\begin{align}
	0
	&\leq
	\log_2\det(\mathbf{I}+\mathbf{A})
	-
	\log_2\det(\mathbf{I}+\mathbf{B})
	\nonumber\\
	&=
	\sum_{n=0}^{\infty}
	\left[
	\log_2\!\left(1+\lambda_n(\mathbf{A})\right)
	-
	\log_2\!\left(1+\lambda_n(\mathbf{B})\right)
	\right]
	\nonumber\\
	&\leq
	\frac{1}{\ln2}
	\sum_{n=0}^{\infty}
	\left[
	\lambda_n(\mathbf{A})
	-
	\lambda_n(\mathbf{B})
	\right] =
	\frac{
		\operatorname{Tr}(\mathbf{A}-\mathbf{B})
	}{\ln2}.
	\label{eq:logdet_trace_bound}
\end{align}

\subsubsection{Receive-Side Truncation}

We now compare the continuous channel $\mathbf H$ and an
auxiliary receive-truncated channel for a fixed
$\widetilde{\mathbf Q}$. Recall that $\mathbf H$ is given in
\eqref{eq:H_factored}. For the purpose of the comparison, define
the receive-truncated random operator $\mathbf H_{r,N}$ through
its matrix elements as
$$
	[\mathbf H_{r,N}]_{n,m}
	\triangleq
	\sqrt{\mu_n(\mathbf M_N)}
	\,w_{n,m}\,
	\sqrt{\mu_m(\mathcal P)},
	\qquad n,m\geq0,
	\label{eq:receive_truncated_channel}
$$
where the convention in \eqref{eq:MN_zero_extension} is used.
For notational simplicity, we further define
$$
	\mathbf K
	\triangleq
	\frac{1}{\sigma_z}
	\mathbf H\widetilde{\mathbf Q}^{1/2},
	\qquad
	\mathbf K_{r,N}
	\triangleq
	\frac{1}{\sigma_z}
	\mathbf H_{r,N}\widetilde{\mathbf Q}^{1/2}.
	\label{eq:weighted_receive_channels}
$$
Their matrix elements are therefore
\begin{equation}
	[\mathbf K]_{n,m}
	\triangleq
	\frac{1}{\sigma_z}
	\sqrt{\mu_n(\mathcal P)}
	\,w_{n,m}\,
	\sqrt{\mu_m(\mathcal P)q_m},
	\label{eq:K_element}
\end{equation}
\begin{equation}
	[\mathbf K_{r,N}]_{n,m}
	\triangleq
	\frac{1}{\sigma_z}
	\sqrt{\mu_n(\mathbf M_N)}
	\,w_{n,m}\,
	\sqrt{\mu_m(\mathcal P)q_m}.
	\label{eq:KrN_element}
\end{equation}
Note that the same realization of the i.i.d. Gaussian
coefficients $\{w_{n,m}\}$ is utilized for both
$\mathbf K$ and $\mathbf K_{r,N}$.

By the Hilbert--Schmidt well-posedness established in
Appendix~\ref{appendix:channel_representations}, $\mathbf K$ is Hilbert--Schmidt almost surely.
Using $\mathbb E[|w_{n,m}|^2]=1$, we obtain
\begin{align}
	\mathbb{E}
	\left[
	\|\mathbf{K}\|_{\mathrm{HS}}^2
	\right] &= \mathbb{E}\left[\sum_{n=0}^{\infty}\sum_{m=0}^{\infty}\left|\mathbf{K}_{n,m}\right|^2\right] =\sum_{n=0}^{\infty}\sum_{m=0}^{\infty}\mathbb{E}\left[\left|\mathbf{K}_{n,m}\right|^2\right] \nonumber \\
	&=
	\frac{1}{\sigma_z^2}
	\left(
	\sum_{n=0}^{\infty}\mu_n(\mathcal{P})
	\right)
	\left(
	\sum_{m=0}^{\infty}
	\mu_m(\mathcal{P})q_m
	\right)
	\nonumber\\
	&\leq
	\rho\mu_0(\mathcal{P})
	\operatorname{Tr}(\mathcal{P})
	<\infty,
	\label{eq:K_HS_full}
\end{align}
where $\rho=P_T/\sigma_z^2$, and the last inequality follows
from \eqref{eq:appendixF_power} and the non-increasing ordering
of $\mu_m(\mathcal P)$. Specifically, $\sum_{m=0}^{\infty}
	\mu_m(\mathcal P)q_m
	\leq
	\mu_0(\mathcal P)
	\sum_{m=0}^{\infty}q_m
	\leq
	\mu_0(\mathcal P)P_T$.
The second equality in \eqref{eq:K_HS_full} follows from
Tonelli's theorem because
$|[\mathbf K]_{n,m}|^2\geq0$.
Similarly,
$$
	\mathbb E
	\left[
	\|\mathbf K_{r,N}\|_{\mathrm{HS}}^2
	\right]
	\leq
	\frac{\mu_0(\mathcal P)P_T}{\sigma_z^2}
	\operatorname{Tr}(\mathbf M_N)
	<\infty.
	\label{eq:KrN_HS_bound}
$$
Hence, $\mathbf K$ and $\mathbf K_{r,N}$ are Hilbert--Schmidt
operators almost surely. 
Consequently,
$\mathbf K\mathbf K^\dagger$,
$\mathbf K^\dagger\mathbf K$,
$\mathbf K_{r,N}\mathbf K_{r,N}^\dagger$, and
$\mathbf K_{r,N}^\dagger\mathbf K_{r,N}$
are positive trace-class operators.

Since $\mathbf K\mathbf K^\dagger$ and
$\mathbf K^\dagger\mathbf K$ have the same non-zero
eigenvalues, their Fredholm determinants are equal:
\begin{equation}
	\det
	\left(
	\mathbf I+\mathbf K\mathbf K^\dagger
	\right)
	=
	\det
	\left(
	\mathbf I+\mathbf K^\dagger\mathbf K
	\right).
	\label{eq:K_fredholm_identity}
\end{equation}
Likewise,
\begin{equation}
	\det
	\left(
	\mathbf I+\mathbf K_{r,N}\mathbf K_{r,N}^\dagger
	\right)
	=
	\det
	\left(
	\mathbf I+\mathbf K_{r,N}^\dagger\mathbf K_{r,N}
	\right).
	\label{eq:KrN_fredholm_identity}
\end{equation}

To compare the two Gram operators (i.e., $\mathbf{K}^{\dagger}\mathbf{K}$ and $\mathbf{K}_{r, N}^{\dagger}\mathbf{K}_{r,N}$), define
$\mathbf{L}_{r,N}$ by
\begin{equation}
	[\mathbf L_{r,N}]_{n,m}
	\triangleq
	\frac{1}{\sigma_z}
	\sqrt{
		\mu_n(\mathcal P)-\mu_n(\mathbf M_N)
	}
	\,w_{n,m}\,
	\sqrt{\mu_m(\mathcal P)q_m}.
	\label{eq:LrN_element}
\end{equation}
The square root is well defined by
\eqref{eq:capacity_spectral_ordering}. Moreover, following the
same calculation as in \eqref{eq:K_HS_full}, we have
\begin{equation}
	\mathbb E
	\left[
	\|\mathbf L_{r,N}\|_{\mathrm{HS}}^2
	\right]
	\leq
	\rho\mu_0(\mathcal P)
	\left[
	\operatorname{Tr}(\mathcal P)
	-
	\operatorname{Tr}(\mathbf M_N)
	\right]
	<\infty.
	\label{eq:LrN_HS_bound}
\end{equation}
Thus, $\mathbf L_{r,N}$ is Hilbert--Schmidt almost surely.
Because $\mathbf K$, $\mathbf K_{r,N}$, and
$\mathbf L_{r,N}$ are constructed using the same Gaussian
coefficients $\{w_{n,m}\}$, their Gram operators satisfy
\begin{equation}
	\mathbf K^\dagger\mathbf K
	-
	\mathbf K_{r,N}^\dagger\mathbf K_{r,N}
	=
	\mathbf L_{r,N}^\dagger\mathbf L_{r,N}
	\succeq0.
	\label{eq:receive_gram_identity}
\end{equation}
The identity follows directly by comparing the matrix elements
in \eqref{eq:K_element}, \eqref{eq:KrN_element}, and
\eqref{eq:LrN_element}.

Given the definition above, the intermediate receive-truncated
ergodic rate is defined as
$$
	\mathcal I_{r,N}
	\left(
	\widetilde{\mathbf Q}
	\right)
	\triangleq
	\mathbb E
	\left[
	\log_2
	\det
	\left(
	\mathbf I+
	\mathbf K_{r,N}
	\mathbf K_{r,N}^\dagger
	\right)
	\right].
	\label{eq:IrN_definition}
$$
Applying \eqref{eq:logdet_trace_bound} to
\eqref{eq:receive_gram_identity}, together with
\eqref{eq:K_fredholm_identity} and
\eqref{eq:KrN_fredholm_identity}, yields
$$
	0
	\leq
	\mathcal I
	\left(
	\widetilde{\mathbf Q}
	\right)
	-
	\mathcal I_{r,N}
	\left(
	\widetilde{\mathbf Q}
	\right) \leq
	\frac{1}{\ln2}
	\mathbb E
	\left[
	\operatorname{Tr}
	\left(
	\mathbf L_{r,N}^\dagger
	\mathbf L_{r,N}
	\right)
	\right].
	\label{eq:receive_rate_gap_1}
$$
Using
$\operatorname{Tr}
(\mathbf L_{r,N}^\dagger\mathbf L_{r,N})
=
\|\mathbf L_{r,N}\|_{\mathrm{HS}}^2$
and \eqref{eq:LrN_HS_bound}, we finally obtain
\begin{equation}
	0
	\leq
	\mathcal I
	\left(
	\widetilde{\mathbf Q}
	\right)
	-
	\mathcal I_{r,N}
	\left(
	\widetilde{\mathbf Q}
	\right)
	\leq
	\frac{\rho\mu_0(\mathcal P)}{\ln2}
	\left[
	\operatorname{Tr}(\mathcal P)
	-
	\operatorname{Tr}(\mathbf M_N)
	\right].
	\label{eq:receive_rate_gap}
\end{equation}

\subsubsection{Transmit-Side Truncation}

We next retain the receive-side truncated spectrum and truncate
the transmit-side spectrum from
$\{\mu_m(\mathcal P)\}_{m\geq0}$ to
$\{\mu_m(\mathbf M_N)\}_{m\geq0}$, where the zero-extension
convention in \eqref{eq:MN_zero_extension} is used.
For the common admissible power-allocation sequence
$\{q_m\}_{m\geq0}$ introduced in
\eqref{eq:appendixF_power}, define the finite-dimensional
diagonal covariance matrix
\begin{equation}
	\widetilde{\mathbf Q}_N
	\triangleq
	\operatorname{diag}
	(q_0,q_1,\ldots,q_{N-1}).
	\label{eq:QN_from_common_sequence}
\end{equation}
Clearly,
$\widetilde{\mathbf Q}_N\succeq0$ and
$\operatorname{Tr}(\widetilde{\mathbf Q}_N)\leq P_T$.

Using the same Gaussian coefficients $\{w_{n,m}\}$, we define
the fully truncated weighted random operator $\mathbf K_N$
through its elements as
\begin{equation}
	[\mathbf K_N]_{n,m}
	\triangleq
	\frac{1}{\sigma_z}
	\sqrt{\mu_n(\mathbf M_N)}
	\,w_{n,m}\,
	\sqrt{\mu_m(\mathbf M_N)q_m},
	\;\; n,m\geq0.
	\label{eq:KN_element}
\end{equation}
By the zero-extension convention in
\eqref{eq:MN_zero_extension}, $\mathbf K_N$ has nonzero entries
only in its leading $N\times N$ block. Using the identification
of $\mathbf H_{w,N}$ with the leading $N\times N$ block of the
Gaussian array $\{w_{n,m}\}$, this leading block is precisely
\begin{equation}
	\frac{1}{\sigma_z}
	\mathbf H_N
	\widetilde{\mathbf Q}_N^{1/2}.
	\label{eq:KN_leading_block}
\end{equation}
Note that $\mathbf K_N$ is only a
zero-padded auxiliary embedding of the finite-dimensional
weighted truncated channel.

Since $\mathbf K_N$ has finite rank, it is Hilbert--Schmidt
almost surely. Define the transmit-side perturbation operator
$\mathbf L_{t,N}$ by
\begin{gather}
	[\mathbf L_{t,N}]_{n,m}
	\triangleq
	\frac{1}{\sigma_z}
	\sqrt{\mu_n(\mathbf M_N)}
	\,w_{n,m}\,
	\sqrt{
		\left[
		\mu_m(\mathcal P)-\mu_m(\mathbf M_N)
		\right]q_m
	}, \nonumber \\
	n,m\geq0.
	\label{eq:LtN_element}
\end{gather}
Again, the square root is well defined by
\eqref{eq:capacity_spectral_ordering}. Moreover,
\begin{align}
	\mathbb E
	\left[
	\|\mathbf L_{t,N}\|_{\mathrm{HS}}^2
	\right]
	=
	\frac{\operatorname{Tr}(\mathbf M_N)}
	{\sigma_z^2}
	\sum_{m=0}^{\infty}
	\left[
	\mu_m(\mathcal P)-\mu_m(\mathbf M_N)
	\right]q_m
	<\infty .
	\label{eq:LtN_HS}
\end{align}
Hence, $\mathbf L_{t,N}$ is Hilbert--Schmidt almost surely. Following directly from
\eqref{eq:KrN_element},
\eqref{eq:KN_element}, and
\eqref{eq:LtN_element}, the corresponding Gram operators
satisfy the exact identity
\begin{equation}
	\mathbf K_{r,N}\mathbf K_{r,N}^{\dagger}
	-
	\mathbf K_N\mathbf K_N^{\dagger}
	=
	\mathbf L_{t,N}\mathbf L_{t,N}^{\dagger}
	\succeq0 .
	\label{eq:transmit_gram_identity}
\end{equation}

We next connect the auxiliary operator $\mathbf K_N$ to the
finite-dimensional rate functional
$\mathcal I_N(\cdot)$ defined in~\eqref{eq:finite_capacity_diagonal_Q}. Since
$\mathbf K_N$ is supported only on its leading $N\times N$
block, its Fredholm determinant reduces exactly to the
corresponding finite-dimensional determinant. Therefore,
using \eqref{eq:KN_leading_block},
\begin{align}
	&\mathbb E
	\left[
	\log_2
	\det
	\left(
	\mathbf I+
	\mathbf K_N\mathbf K_N^\dagger
	\right)
	\right]
	\nonumber\\
	&=
	\mathbb E
	\left[
	\log_2
	\det
	\left(
	\mathbf I_N+
	\frac{1}{\sigma_z^2}
	\mathbf H_N
	\widetilde{\mathbf Q}_N
	\mathbf H_N^\dagger
	\right)
	\right] =
	\mathcal I_N
	\left(
	\widetilde{\mathbf Q}_N
	\right).
	\label{eq:KN_finite_rate_identity}
\end{align}

Applying \eqref{eq:logdet_trace_bound} to
\eqref{eq:transmit_gram_identity} and using
\eqref{eq:KN_finite_rate_identity} gives
$$
	0
	\leq
	\mathcal I_{r,N}
	\left(
	\widetilde{\mathbf Q}
	\right)
	-
	\mathcal I_N
	\left(
	\widetilde{\mathbf Q}_N
	\right)
	\leq
	\frac{1}{\ln2}
	\mathbb E
	\left[
	\operatorname{Tr}
	\left(
	\mathbf L_{t,N}\mathbf L_{t,N}^{\dagger}
	\right)
	\right].
	\label{eq:transmit_rate_gap_1}
$$

Since
$q_m\leq\sum_{k=0}^{\infty}q_k\leq P_T$
for every $m\geq0$, we obtain from~\eqref{eq:LtN_HS} that
\begin{align}
	&
	\mathbb E
	\left[
	\operatorname{Tr}
	\left(
	\mathbf L_{t,N}\mathbf L_{t,N}^{\dagger}
	\right)
	\right]
	\leq
	\rho\operatorname{Tr}(\mathbf M_N)
	\sum_{m=0}^{\infty}
	\left[
	\mu_m(\mathcal P)-\mu_m(\mathbf M_N)
	\right]
	\nonumber\\
	&=
	\rho\operatorname{Tr}(\mathbf M_N)
	\left[
	\operatorname{Tr}(\mathcal P)
	-
	\operatorname{Tr}(\mathbf M_N)
	\right], \nonumber
	\label{eq:LtN_trace_bound}
\end{align}
where the last equality follows from
\eqref{eq:trace_equivalence}. Therefore,
\begin{equation}
	0
	\leq
	\mathcal I_{r,N}
	\left(
	\widetilde{\mathbf Q}
	\right)
	-
	\mathcal I_N
	\left(
	\widetilde{\mathbf Q}_N
	\right)
	\leq
	\frac{
		\rho\operatorname{Tr}(\mathbf M_N)
	}{\ln2}
	\left[
	\operatorname{Tr}(\mathcal P)
	-
	\operatorname{Tr}(\mathbf M_N)
	\right].
	\label{eq:transmit_rate_gap}
\end{equation}

\subsubsection{Capacity Optimization and Conclusion}

Combining \eqref{eq:receive_rate_gap} and
\eqref{eq:transmit_rate_gap}, for every admissible
power-allocation sequence $\{q_m\}_{m\geq0}$ satisfying
\eqref{eq:appendixF_power}, we obtain
\begin{align}
	0
	&\leq
	\mathcal I
	\left(
	\widetilde{\mathbf Q}
	\right)
	-
	\mathcal I_N
	\left(
	\widetilde{\mathbf Q}_N
	\right)
	\nonumber\\
	&\leq
	\frac{\rho}{\ln2}
	\left[
	\mu_0(\mathcal P)
	+
	\operatorname{Tr}(\mathbf M_N)
	\right]
	\left[
	\operatorname{Tr}(\mathcal P)
	-
	\operatorname{Tr}(\mathbf M_N)
	\right].
	\label{eq:fixed_Q_capacity_gap}
\end{align}

We now relate the fixed-allocation comparison above to the two
capacity optimizations. By construction, every admissible sequence in \eqref{eq:appendixF_power} induces the continuous covariance \(\widetilde{\mathbf Q}\) and its finite-dimensional restriction \(\widetilde{\mathbf Q}_N\) in \eqref{eq:QN_from_common_sequence}; conversely, every feasible \(N\)-dimensional diagonal covariance admits a zero extension satisfying \eqref{eq:appendixF_power}.

According to the definitions of $C_{\mathrm erg}$ and $C_{\mathrm erg}^{(N)}$ given in Lemma~\ref{lemma:continuous_capacity} and Lemma~\ref{lemma:truncated_capacity}, respectively, their capacity gap is given by
\begin{align}
	C_{\rm erg}
	-
	C_{\rm erg}^{(N)}
	&=
	\sup_{\{q_m\}}
	\mathcal I
	\left(
	\widetilde{\mathbf Q}
	\right)
	-
	\sup_{\{q_m\}}
	\mathcal I_N
	\left(
	\widetilde{\mathbf Q}_N
	\right)
	\nonumber\\
	&\leq
	\sup_{\{q_m\}}
	\left[
	\mathcal I
	\left(
	\widetilde{\mathbf Q}
	\right)
	-
	\mathcal I_N
	\left(
	\widetilde{\mathbf Q}_N
	\right)
	\right],
	\label{eq:capacity_sup_gap}
\end{align}
where all suprema in \eqref{eq:capacity_sup_gap} are taken over
the admissible power-allocation sequences satisfying
\eqref{eq:appendixF_power}. Note that the continuous and
truncated channels need not share the same capacity-achieving
power allocation; the inequality above only requires the two
optimizations to be parameterized over the same admissible
power-allocation set.

Since \eqref{eq:fixed_Q_capacity_gap} holds for every admissible power-allocation sequence, it also implies $C_{\rm erg}\ge C_{\rm erg}^{(N)}$. Combining \eqref{eq:fixed_Q_capacity_gap} and
\eqref{eq:capacity_sup_gap}, we obtain~\eqref{eq:capacity_trace_defect_tight} in
Theorem~\ref{theorem:capacity_convergence}. Finally, since
$\mu_0(\mathcal P)\leq1$ and
$\operatorname{Tr}(\mathbf M_N)
\leq
\operatorname{Tr}(\mathcal P)=c^2/\pi$, we obtain the bound in~\eqref{eq:capacity_trace_defect}.
When $N_{1D}>ec/4$, applying the upper bound in Theorem~\ref{theorem:approx_error} for $\operatorname{Tr}(\mathcal{P})
-
\operatorname{Tr}(\mathbf{M}_N)$ in \eqref{eq:capacity_trace_defect} further yields the explicit bound in~\eqref{eq:capacity_convergence}, which completes the proof.

\section{}
\begin{lemma}
	For any real parameter $A > 0$ and integer order $m \ge 1$, the magnitude of the Bessel function of the first kind $J_m(A)$ is strictly bounded by the following super-exponential envelope:
	\begin{equation}
		\left|J_m(A)\right| < \frac{1}{\sqrt{2\pi m}} \left( \frac{e A}{2 m} \right)^m.
		\label{eq:Bessel_inequality_1}
	\end{equation}
	\label{lemma:Bessel_inequality}
\end{lemma}
\begin{proof}
	We first invoke Poisson's integral representation of the Bessel function~\cite[Eq.9.1.20]{Abramowitz1964Handbook}
	$$
		J_m(A) = \frac{(A/2)^m}{\sqrt{\pi} \Gamma(m + 1/2)} \int_0^\pi \cos(A \cos \theta) \sin^{2m} \theta \, d\theta.
	$$
	Taking the absolute value on both sides and applying the triangle inequality for integrals, we bound the integrand by exploiting $\left|\cos(A \cos \theta)\right| \le 1$ for all real $A$ and $\theta$:
	\begin{equation}
		\left|J_m(A)\right| \le \frac{(A/2)^m}{\sqrt{\pi} \Gamma(m + 1/2)} \int_0^\pi \sin^{2m} \theta \, d\theta.
		\label{eq:inequality_Bessel}
	\end{equation}
	The remaining definite integral is a standard trigonometric identity that evaluates exactly to Gamma functions:
	\begin{equation}
		\int_0^\pi \sin^{2m} \theta \, d\theta = \frac{\sqrt{\pi} \Gamma(m + 1/2)}{\Gamma(m + 1)}.
		\label{eq:trigono_Gamma}
	\end{equation}
	Substituting (\ref{eq:trigono_Gamma}) into (\ref{eq:inequality_Bessel}), we obtain:
	\begin{equation}
		\left|J_m(A)\right| \le \frac{(A/2)^m}{\Gamma(m + 1)} = \frac{(A/2)^m}{m!}.
		\label{eq:Bessel_inequality_2}
	\end{equation}
	To formulate this factorial envelope into a tractable closed-form, we employ the strict lower bound of Stirling's approximation established by Robbins~\cite{robbins1955remark}: $m! > \sqrt{2\pi m} \left( \frac{m}{e} \right)^m$.
	Substituting Robbins' bound into (\ref{eq:Bessel_inequality_2}), we obtain the strict super-exponential envelope in (\ref{eq:Bessel_inequality_1}).
	This completes the proof.
\end{proof}
	
\section{}
\label{appendix:quadrature_nodes}
To determine the convergence thresholds, we map the band-limited physical properties of the 1D PSWFs into the 2D polar coordinate system to evaluate their maximum spatial frequencies, or effective exponential types.

\subsubsection{Radial quadrature threshold} We first analyze the radial integrand, given by
\begin{equation}
	F_r(r) = \phi_p(\frac{r\cos\theta}{c})\phi_q(\frac{r\sin\theta}{c})\phi_j(\frac{r\cos\theta}{c})\phi_\ell(\frac{r\sin\theta}{c})r.
	\label{eq:F_r_r}
\end{equation}
Recall from \eqref{eq:Fourier_PSWF} that the Fourier transform of a 1D PSWF $\phi_n(t)$ is strictly supported within the physical bandwidth $[-c, c]$. Specifically, $\phi_n(\frac{r\cos\theta}{c})$ can be expressed as:
\begin{equation}
	\phi_n\Big(\frac{r\cos\theta}{c}\Big) = \frac{1}{2\pi}\int_{-c}^{c} S_n(k) e^{i k \frac{r\cos\theta}{c}} \,dk,
	\label{eq:phi_n_Fourier}
\end{equation}
where $S_n(k) = \frac{2\pi}{c \lambda_n}\phi_n(\frac{k}{c})$.
From (\ref{eq:phi_n_Fourier}), $\phi_n\Big(\frac{r\cos\theta}{c}\Big) $ can be viewed as a continuous superposition of complex exponential functions $e^{i\frac{k\cos\theta}{c}r}$. 
For a given angle $\theta$, and noting that $k\in[-c, c]$, the maximum spatial frequency with respect to the radial variable $r$ is bounded by $\left|\cos\theta\right|$.
Similarly, the maximum spatial frequency for the terms containing $\sin\theta$ is $|\sin\theta|$.

By the Paley-Wiener theorem~\cite{rudin1974real}, each term in (\ref{eq:F_r_r}) is an entire function of exponential type $|\cos\theta|$ or $|\sin\theta|$. Since the polynomial multiplier $r$ is of degree 1, it does not alter the overall exponential type.
By the properties of entire functions, the exponential type of a product of entire functions is at most the sum of their individual exponential types, which corresponds to the convolution of their bounded spectral supports.
Therefore, the exponential type of $F_r(r)$ is upper bounded by:
$$
	B_r(\theta) = 2(\left|\cos\theta\right| + \left|\sin\theta\right|).
	\label{eq:total_bandwidth}
$$
To guarantee convergence across all integration angles, we evaluate the global maximum of $B_r(\theta)$, which occurs at odd integer multiples of $\frac{\pi}{4}$,
yielding a uniform upper bound of $2\sqrt{2}$ on the exponential type. 

To apply the GLQ method, the physical radial interval $r \in [0, c]$ is affinely mapped to the standard GLQ interval $s \in [-1, 1]$ via the transformation $r = \frac{c}{2}(s+1)$. This mapping scales the uniform upper bound on the exponential type with respect to $s$ to:
$$
	\Omega_s = 2\sqrt{2} \cdot \frac{c}{2} = \sqrt{2}c.
$$
According to~\cite[Theorem 1]{Yan2026Computational}, for an entire function of exponential type $\sqrt{2}c$, its quadrature error is globally upper-bounded by $C(M_r)\left(\frac{e\sqrt{2}c}{4M_r}\right)^{2M_r}$, where $C(M_r)$ is an algebraically decaying prefactor. Consequently, the super-exponential convergence regime is reached when the base of the error envelope is less than one (i.e., $\frac{e\sqrt{2}c}{4M_r} < 1$), which yields the threshold:
$$
	M_r > \frac{e\sqrt{2}c}{4}.
	\label{eq:M_r_threshold}
$$
	
\subsubsection{Angular quadrature threshold}
For the angular dimension $\theta \in [0, 2\pi]$, we analyze the integrand as a periodic function, denoted by
$$
	F_\theta(\theta) = \phi_p(\frac{r\cos\theta}{c})\phi_q(\frac{r\sin\theta}{c})\phi_j(\frac{r\cos\theta}{c})\phi_\ell(\frac{r\sin\theta}{c}). 
$$
Using the continuous spectral representation of the 1D PSWFs, this product can be formulated as a quadruple integral over the four-dimensional wavenumber domain $[-c, c]^4$:
\begin{align}
	F_\theta(\theta) = \iiiint_{[-c,c]^4}& S_p(k_1)S_q(k_2)S_j(k_3)S_\ell(k_4) \nonumber \\ 
	&\times e^{i \Phi(r, \theta, \mathbf{k})} \,dk_1 dk_2 dk_3 dk_4,
	\label{eq:def_F_theta}
\end{align}
where the composite phase function $\Phi(r, \theta, \mathbf{k})$ is a linear superposition of the individual exponential arguments:
$$
	\Phi(r, \theta, \mathbf{k}) = \frac{r}{c} (k_1\cos\theta + k_2\sin\theta + k_3\cos\theta + k_4\sin\theta).
$$
We regroup the angular variables as $K_x = k_1 + k_3$ and $K_y = k_2 + k_4$. Since the spectral support of each PSWF is bounded by $k_i \in [-c, c]$, the regrouped terms $K_x$ and $K_y$ are confined to the expanded domains $K_x \in [-2c, 2c]$ and $K_y \in [-2c, 2c]$. Through harmonic addition, the phase function simplifies to:
\begin{align}
	\Phi(r, \theta,K_x, K_y) & = \frac{r}{c} (K_x\cos\theta + K_y\sin\theta) \nonumber \\
	& = \frac{r}{c} \sqrt{K_x^2 + K_y^2} \cos(\theta - \alpha),
	\label{eq:Phi_phase}
\end{align} 
where $\alpha$ is a phase angle determined by $K_x$ and $K_y$.
The maximum possible phase amplitude, denoted by $A_{\max}$, represents the worst-case angular fluctuation of the integrand across all possible wavenumber combinations $(K_x, K_y) \in [-2c, 2c]^2$ and all radial distances $r \in [0, c]$. Based on the coefficient of the cosine term in (\ref{eq:Phi_phase}), this maximum is achieved when $r=c$, $K_x = \pm 2c$, and $K_y = \pm 2c$, yielding:
\begin{equation}
	A_{\max} = \frac{c}{c}\sqrt{(\pm 2c)^2 + (\pm 2c)^2} = 2\sqrt{2}c.
	\label{eq:A_max}
\end{equation}  
Because $F_\theta(\theta)$ in (\ref{eq:def_F_theta}) is a continuous superposition of $e^{i\Phi(r, \theta, \mathbf{k})}$, the decay of its angular Fourier coefficients can be controlled using the worst-case phase amplitude $A_{\max}$. To quantify this decay, we invoke the classic Jacobi-Anger expansion~\cite{Abramowitz1964Handbook}:  
$$
	e^{i A_{\max} \cos(\theta - \alpha)} = \sum_{m=-\infty}^{\infty} i^m J_m(A_{\max}) e^{i m (\theta - \alpha)},
$$
where $J_m(\cdot)$ denotes the Bessel function of the first kind of integer order $m$. 

Given the strict upper bound established in Lemma \ref{lemma:Bessel_inequality}, $J_m(A_{\max})$ enters the super-exponential decay regime when $\frac{eA_{\max}}{2m} < 1$, which corresponds to $m > \frac{eA_{\max}}{2}$. Since $A_{\max} = 2\sqrt{2}c$ as given in (\ref{eq:A_max}), we obtain the super-exponential decay threshold as $m > \sqrt{2}ec$. Therefore, $e^{iA_{\max}\cos(\theta-\alpha)}$, and consequently $F_{\theta}(\theta)$, have Fourier coefficients that enter the super-exponential decay regime beyond this threshold.

We evaluate the $2\pi$-periodic function $F_\theta(\theta)$ using the trapezoidal rule on a uniform grid, which provides exponential convergence~\cite{Trefethen2014Exponentially}.
Under this rule, the quadrature error is dominated by the aliasing of Fourier coefficients beyond the sampling rate. 
Thus, choosing $M_{\theta}$ beyond the same threshold ensures that the aliased Fourier coefficients lie in the super-exponentially decaying tail identified above.
This yields the following sufficient quadrature threshold:
$$
	M_\theta > \sqrt{2}ec.
$$

\bibliographystyle{IEEEtran}
\bibliography{mybibliography.bib}

\end{document}